\ifx\synctex\undefined\else\synctex=1\fi
\documentclass[preprint,numbers]{sigplanconf}
\usepackage{hyperref}
\usepackage{mathpartir}
\usepackage{microtype}
\usepackage{amsmath,amssymb,amsthm,alltt}
\usepackage{mathtools,paralist}
\usepackage{thmtools,thm-restate}
\usepackage{todonotes}
\usepackage{tikz}
\usetikzlibrary{decorations.pathreplacing}

\hypersetup{
    unicode=false,          % non-Latin characters in Acrobat’s bookmarks
    pdftoolbar=true,        % show Acrobat’s toolbar?
    pdfmenubar=true,        % show Acrobat’s menu?
    pdffitwindow=false,     % window fit to page when opened
    pdfstartview={FitH},    % fits the width of the page to the window
    pdftitle={My title},    % title
    pdfauthor={Author},     % author
    pdfsubject={Subject},   % subject of the document
    pdfcreator={Creator},   % creator of the document
    pdfproducer={Producer}, % producer of the document
    pdfnewwindow=true,      % links in new window
    colorlinks=true,        % false: boxed links; true: colored links
    linkcolor=black,        % color of internal links
    citecolor=black,        % color of links to bibliography
    filecolor=magenta,      % color of file links
   urlcolor=cyan           % color of external links
}

\newtheorem{theorem}{Theorem}[section]
\newtheorem{lemma}[theorem]{Lemma}
\newtheorem{corollary}[theorem]{Corollary}
\newtheorem{definition}{Definition}[section]
\newtheorem{example}{Example}
\newtheorem*{claim}{Claim}

\newcommand{\ignore}[1]{}

\newcommand{\sstree}[1]{\tikz[level distance=3ex,baseline=-1ex]{#1}}
\newcommand{\stree}[2]{\sstree{\node{$\bullet^2$} child{node{$#1$}} child{node{$#2$}};}}
\newcommand{\down}{\!\downarrow}

\newcommand{\Mm}{\mathcal{M}}
\newcommand{\es}{\emptyset}
\newcommand{\set}[1]{\{#1\}}
\renewcommand{\S}{\Sigma}
\newcommand{\Aa}{\mathcal A}
\newcommand{\Bb}{\mathcal B}
\newcommand{\Cc}{\mathcal C}
\newcommand{\Hh}{\mathcal H}
\newcommand{\Pp}{\mathcal P}
\newcommand{\Ss}{\mathcal S}
\newcommand{\Rr}{\mathcal R}
\newcommand{\Ll}{\mathcal L}
\newcommand{\Tt}{\mathcal T}
\newcommand{\LT}{\Ll\Tt}
\newcommand{\ord}{\mathit{ord}}
\newcommand{\Unb}{\mathit{Diag}}
\newcommand{\restr}{{\restriction}}
\newcommand{\lista}{\mathit{list}}
\newcommand{\merge}{\mathit{merge}}

\newcommand{\Vars}[1]{\overline{\mathbf{x}_#1}} %\mathit{Vars}_#1}

\newcommand{\tr}{\mathit{tr}}
\newcommand{\trcum}{\mathit{tr}_\mathit{cum}}
\newcommand{\Nat}{\mathbb{N}}
\renewcommand{\r}{\mathbf r}

\newcommand{\goestoA}[2]{\stackrel {#2} {\longrightarrow}_{#1}}
\newcommand{\goesto}[1]{\goestoA {} {#1}}
\newcommand{\lang}[1]{\mathcal{L}(#1)}
\newcommand{\trans}[1]{\mathcal{T}(#1)}

\newcommand{\sem}[1]{\lbrack\!\lbrack #1 \rbrack\!\rbrack}%
\begin{document}
	
\title{The Diagonal Problem for Higher-Order Recursion Schemes is Decidable}

\authorinfo{Lorenzo~Clemente \thanks{This work was partially supported by the Polish National Science Centre grant 2013/09/B/ST6/01575.}
\and Pawe\l~Parys\thanks{This work was partially supported by the National Science Center (decision DEC-2012/07/D/ST6/02443).}}
           {University of Warsaw\\Warsaw, Poland}
           {\{l.clemente,parys\}@mimuw.edu.pl}
\authorinfo{Sylvain~Salvati\and Igor~Walukiewicz}
           {University of Bordeaux, CNRS, INRIA\\Bordeaux, France}
           {\{sylvain.salvati,igw\}@labri.fr}

\maketitle

% \terms
% term1, term2


\begin{abstract}
  A non-deterministic recursion scheme recognizes a language of finite trees.
  This very expressive model can simulate, among others, higher-order pushdown automata with collapse.
  We show decidability of the diagonal problem for schemes.
  This result has several interesting consequences.
  In particular, it gives an algorithm that computes the downward closure of languages of words recognized by schemes.
  In turn, this has immediate application to separability problems and reachability analysis of concurrent systems.
\end{abstract}

%%% Local Variables:
%%% mode: latex
%%% TeX-master: "main"
%%% End:

\section{Introduction}

The \emph{diagonal problem} is a decision problem with a number of
interesting algorithmic consequences. It is a central subproblem for
computing the downward closure of languages of words~\cite{Zetzsche:ICALP:2015}, as well as for the
problem of separability by piecewise-testable languages~\cite{CzerwinskiMartensRooijenZeitounZetzsche}. It is
used in deciding reachability of a certain type of parameterized concurrent
systems~\cite{LaTorreMuschollWalukiewicz:2015}.
In its original formulation over finite words, the problem asks, for a given set of letters $\S$
and a given language of words $L$, whether for every number $n$ there is a word in $L$ where
every letter from $\S$ occurs at least $n$ times.
In this paper, we study a generalization of the diagonal problem for languages of finite trees recognized by non-deterministic higher-order recursion schemes.

\emph{Higher-order recursion schemes} are algorithmically manageable
abstractions of higher-order programs. Higher-order features are now
present in most mainstream languages like Java, JavaScript, Python, or
C++. Higher-order schemes, or, equivalently, simply typed
lambda-calculus with a fixpoint combinator, are a formalism that can
faithfully model the control flow in higher-order programs. In this paper, we
consider non-deterministic higher-order recursion schemes as recognizers of languages of finite trees.
In other words we consider higher-order OI grammars~\cite{io_oi_damm,kobele2015}.
This is an expressive formalism covering many other models such as
indexed grammars \cite{Aho:JACM:1968},
ordered multi-pushdown automata \cite{BreveglieriCherubiniCitriniCrespi-Reghizzi:Ordered:1996},
or the more general higher-order pushdown automata with collapse \cite{HagueMurawskiOngSerre:Collapsible:2008}
(cf. also the equivalent model of ordered tree-pushdown automata \cite{ClementeParysSalvatiWalukiewicz:FSTTCS:2015}).

Our main result is a procedure for solving the diagonal problem for higher-order schemes.
This is a missing ingredient to obtain several new decidability results for this model.
It is well-known that schemes have a decidable emptiness problem~\cite{Ong:LICS:2006},
and it can be shown that they are closed under rational linear transductions,
and in particular they form a full trio when restricted to finite word languages.
In this context, a result by Zetzsche~\cite{Zetzsche:ICALP:2015}
entails computability of the \emph{downward closure} of languages of words recognized by higher-order schemes.
Moreover, a recent result by Czerwi\'nski, Martens, van Rooijen, and Zeitoun~\cite{CzerwinskiMartensVanRooijenZeitoun:2015}
entails that the \emph{separability by piecewise testable
  languages} is decidable for languages recognized by higher-order schemes.
Finally, a third example comes from La Torre, Muscholl, and Walukiewicz~\cite{LaTorreMuschollWalukiewicz:2015}
showing how to use downward closures
to decide reachability in parameterized asynchronous shared-memory
concurrent systems where every process is a higher-order scheme.
%or, equivalently, a program written as a simply typed lambda-term with
%recursion. 

\looseness=-1
While the examples above show that the diagonal problem is intimately connected to downward closures%
\footnote{In fact, the diagonal problem, separability by piecewise testable languages, and computing the downward closure
are inter-reducible for full trios \cite{CzerwinskiMartensRooijenZeitounZetzsche}.},
the computation of the downward closure is an important problem in its own right.
The downward closure of a language offers an effective abstraction thereof.
Since the subword relation is a well quasi-order \cite{Higman:1952}, the downward
closure of a language is always a regular language determined by a finite set of
forbidden patterns. %\cite{Haines:1969}
 This abstraction is thus particularly interesting
for complex languages, like those not having a semilinear Parikh
image. While the downward closure is always regular, it is not always
possible to effectively construct a finite automaton for it.
This is obviously the case for classes with undecidable emptiness (since the downward closure preserves emptiness),
but it is also the case for relatively better behaved classes for which the emptiness problem is decidable,
such as Church-Rosser languages \cite{GruberHozerKutrib:TCS:2007}, and lossy channel systems \cite{Mayr:TCS:2003}.

The problem of computing the downward closure of a language has
attracted a considerable attention recently. Early results show how
to compute it for context-free languages~\cite{vanLeeuwen:1978,Courcelle:1991} (cf. also \cite{BachmeierLuttenbergerSchlund:2015}),
for Petri-net languages~\cite{DownwardPN:ICALP:2010},
for stacked counter automata \cite{Zetzsche:STACS:2015},
and context-free FIFO rewriting systems and $0L$-systems \cite{AbdullaBoassonBouajjani:2001}.
More recently, Zetzsche~\cite{Zetzsche:ICALP:2015} has given an
algorithm for indexed grammars, or equivalently for second-order
pushdown automata. Hague, Kochems, and
Ong~\cite{DBLP:conf/popl/HagueKO16} have made an important
further advance by showing how to compute the downward closure of the 
language of pushdown automata of arbitrary order. In this paper, we
complete the picture by giving an algorithm for the more general model of higher-order pushdown automata with collapse \cite{HagueMurawskiOngSerre:Collapsible:2008}.
We use the fact that these automata recognize the same class of languages as higher-order recursion schemes,
and we work with the latter model instead.

Let us briefly outline our approach.
While are mainly interested in higher-order recursion schemes (HORSes) generating finite words,
for technical reasons we also need to consider narrow trees,
i.e., trees with a bounded number of paths.
In this we follow an idea of Hague et al.~\cite{DBLP:conf/popl/HagueKO16} who have used this
technique for higher-order pushdown automata (without collapse). 
For a HORS $\Ss$ and a set of letters $\Sigma$, %let $\Unb_\Sigma(\Ss)$
%denote the property that holds if 
the diagonal problem asks whether
for every $n\in\Nat$ there is a tree
generated by $\Ss$ in which every letter from $\Sigma$ appears at
least $n$ times. Our goal is an algorithm solving this problem. %for deciding whether
%$\Unb_\Sigma(\Ss)$ holds for a given $\Sigma$ and $\Ss$. 
When $\Ss$ is of order $0$, we have a regular grammar,
for which the diagonal problem can be solved by direct inspection.
For higher orders,
apply a transformation that decreases the order by one.
The order is decreased in two steps.
First, we ensure that the HORS generates only narrow trees: we construct
a HORS $\Ss'$, of the same order as $\Ss$, generating only narrow trees
and such that the diagonal problems for $\Ss$ and $\Ss'$ are equivalent. %$\Unb_\Sigma(\Ss)$ holds if and only if $\Unb_{\Sigma'}(\Ss')$
%holds for some $\S'$ that we also construct.
Then, in the narrow HORS $\Ss'$ we lower the order by one: 
we create a HORS $\Ss''$ that is of order smaller by one than $\Ss'$
(but no longer narrow), and such that %$\Unb_{\Sigma'}(\Ss')$ holds
%if and only if $\Unb_{\Sigma''}(\Ss'')$ holds,
%for a $\Sigma''$ that we also construct.
the diagonal problems for $\Ss'$ and $\Ss''$ are equivalent.

While narrowing the HORS is relatively easy to achieve,
the main technical difficulty is order reduction.
This point is probably better explained in terms of
higher-order pushdown automata.
If a higher-order pushdown automaton of order $n$ accepts
with an empty stack then an accepting computation has no choice but
to pop out level-$n$ stacks one by one. In other words, for every
configuration the level-$n$ return points are easily
predictable. Using this we can eliminate them obtaining an automaton
of order $n-1$. When we allow the collapse operation the situation
changes completely: a configuration may have arbitrary many level-$n$
return points, and different computations may use different return points.

In this paper we prefer to use HORSes rather than higher-order pushdown automata with collapse.
Our solution resembles the one from \cite{AsadaKobayashi:ICALP16},
where a word-generating HORS is turned into a tree-generating HORS of order lower by one,
whose frontier language
(the language of words written from left to right in the leaves)
is exactly the language of the original word-generating HORS.
If our narrow trees were of width one (i.e., word-generating),
we could just invoke \cite{AsadaKobayashi:ICALP16},
since their transformation preserves in particular the cardinality of the produced letters.
While in general we need to handle narrow trees instead of words
(a more general input than in \cite{AsadaKobayashi:ICALP16}),
we only prove that our construction preserves the number of their occurrences
(and not their order, thus having a result weaker than in \cite{AsadaKobayashi:ICALP16}).
While the two results are thus formally incomparable,
it is worth remarking that our construction does actually preserve the order of symbols belonging to the same branch of the narrow tree.
%while symbols belonging to different branches are produced in some uncontrolled order.
%

After some preliminaries in Section~\ref{sec:preliminaries},
we state formally our main result and some of its consequences in Section~\ref{sec:result}.
The rest of the paper is devoted to the proof.
In Section~\ref{sec:narrowing}, we present a transformation of a scheme to a narrow one that preserves the order,
and in Section~\ref{sec:lowering} we present the reduction of a narrow scheme to a scheme of a smaller order (but not necessarily narrow).
Both reductions preserve the diagonal problem.
Finally, in Section~\ref{sec:conclusions}, we conclude with some further considerations.

%%% Local Variables:
%%% mode: latex
%%% TeX-master: "main"
%%% End:

\section{Preliminaries}
\label{sec:preliminaries}

\paragraph{Higher-order recursion schemes.}

We use the name ``sort'' instead of ``simple type'' or ``type''
to avoid confusion with the types introduced later.
%, used for describing terms more precisely.
The set of \emph{sorts} is constructed from a unique basic sort $o$ using a binary operation $\to$.
Thus $o$ is a sort, and if $\alpha,\beta$ are sorts, so is $\alpha\to\beta$.
The order of a sort is defined by: $\ord(o)=0$, and $\ord(\alpha\to\beta)=\max(1+\ord(\alpha),\ord(\beta))$.
By convention, $\to$ associates to the right, i.e., $\alpha\to\beta\to\gamma$ is understood as $\alpha\to(\beta\to\gamma)$.
Every sort $\alpha$ can be uniquely written as $\alpha_1\to\alpha_2\to\ldots\to\alpha_n\to o$.
The sort $o\to\dots\to o\to\alpha$ with $r$ occurrences of $o$ is denoted $o^r\to \alpha$, where $o^0\to \alpha$ is simply $\alpha$.

The set of \emph{terms} is defined inductively as follows.
For each sort $\alpha$ there is a countable set of \emph{variables} $x^\alpha,y^\alpha,\dots$ and a countable set of \emph{nonterminals} $A^\alpha,B^\alpha,\dots$; all of them are terms of sort $\alpha$.
There is also a countable set of \emph{letters} $a,b,\dots$; out of a letter $a$ and a sort $\alpha$ of order at most $1$ one can create a \emph{symbol} $a^\alpha$ that is a term of sort $\alpha$.
Moreover, if $K$ and $L$ are terms of sort $\alpha\to\beta$ and
$\alpha$, respectively, then $(K\,L)^\beta$ is a term of sort $\beta$. 
For $\alpha=(o^r\to o)$ we often shorten $a^\alpha$ to $a^r$, and we call $r$ the \emph{rank} of $a^r$.
Moreover, we omit the sort annotation of variables, nonterminals, or terms, 
but note that each of them is implicitly
assigned a particular sort.
We also omit some parentheses when writing terms and denote $(\dots
(K\,L_1) \dots L_n)$ simply by $K L_1\dots L_n$.
A term is called \emph{closed} if it uses no variables.

We deviate here from usual definitions in the detail that letters itself are unranked, and thus out of a single letter $a$ one may create a symbol $a^r$ for every rank $r$.
This is convenient for us, as during the transformations of HORSes described in Sections \ref{sec:narrowing} and \ref{sec:lowering} we need to change the rank of tree nodes, without changing their labels.
Notice, however, that in terms a letter is used always with a particular rank.

A \emph{higher-order recursion scheme} (HORS for
short) is a pair $\Ss=(A_\mathit{init},\Rr)$, where $A_\mathit{init}$
is the \emph{initial nonterminal} that is of sort $o$, 
and $\Rr$ is a finite set of rules of the form $A^\alpha\,x_1^{\alpha_1}\,\dots\,x_k^{\alpha_k}\to K^o$
where $\alpha = \alpha_1\to\dots\to\alpha_k\to o$
and $K$ is a term that uses only variables from the set $\{x_1^{\alpha_1},\dots,x_k^{\alpha_k}\}$.
The order of $\Ss$ is defined as the highest order of a nonterminal for which there is a rule in $\Ss$.
We write $\Rr(\Ss)$ to denote the set of rules of a HORS
$\Ss$. 
Observe that our schemes are \emph{non-deterministic} in the sense that
$\Rr(\Ss)$ can have many rules with the same  
nonterminal on the left side. A scheme with at most one rule for
each nonterminal is called \emph{deterministic}.

Let us now describe the dynamics of HORSes.
Substitution is defined as expected:
\begin{mathpar}
	A[M/x]=A,\and
	a^r[M/x]=a^r,\and
	x[M/x]=M,\and
	y[M/x]=y\mbox{ if }y\neq x,\and
	(K\,L)[M/x]=K[M/x]\,L[M/x].
\end{mathpar}
We shall use the substitution only when $M$ is closed, so there is no need to perform $\alpha$-conversion.
We also allow simultaneous substitutions: we write
$K[M_1/x_1,\dots,M_k/x_k]$ to denote the simultaneous substitution of
$M_1$, \dots, $M_k$ respectively for $x_1$, \dots, $x_k$.
We notice that when the terms $M_i$ are closed, this amounts to apply
the substitutions $[M_i/x_i]$ (with $i\in\{1,\dots,k\}$) in any order.

A HORS $\Ss$  defines a reduction relation $\to_\Ss$ on closed terms:
% as follows:
\begin{mathpar}
	\inferrule{(A\,x_1\,\dots\,x_k\to K)\in\Rr(\Ss)}{A\,M_1\,\dots\,M_k\to_\Ss K[M_1/x_1,\dots,M_k/x_k]}
	\and
	\inferrule{
			K_l\to_\Ss K_l'\mbox{ for some }l\in\{1,\dots,r\}
		\\
			K_i=K_i'\mbox{ for all }i\neq l
		}{
			a^r\,K_1\,\dots\,K_r\to_\Ss a^r\,K_1'\,\dots\,K_r'
		}
\end{mathpar}
We thus apply some of the rules of $\Ss$ to one of the outermost nonterminals in the term.

We are interested in finite trees generated by HORSes.
A closed term $L$ of sort $o$ is a \emph{tree} if it does not contain any nonterminal.
A HORS $\Ss$ \emph{generates} a tree $L$ from a term $K$ if $K\to_\Ss^* L$;
when we do not mention the term $K$ we mean generating from the
initial nonterminal of $\Ss$. Since a scheme may have more than one
rule for some nonterminals, it may generate more than one tree.
We can view a HORS of order $0$ essentially as a finite tree automaton,
thus a HORS of order $0$ generates a regular language of finite trees.

Let $\Delta$ be a finite set of symbols of rank $0$ (called also \emph{nullary} symbols).
A tree $K$ is \emph{$\Delta$-narrow} if it has exactly $|\Delta|$ leaves, each of them labeled by a different symbol from $\Delta$.
A HORS is called \emph{$\Delta$-narrow} if it generates only $\Delta$-narrow trees, and it is called \emph{narrow} if it is $\Delta$-narrow for some $\Delta$.
We are particularly interested in $\Delta$-narrow HORSes for
$|\Delta|=1$; trees generated by them consist of a single branch
and thus can be seen as words.

\paragraph{Transductions.}

A (bottom-up, nondeterministic) \emph{finite tree transducer} (FTT) is a tuple $\Aa = (Q, Q_F, \delta)$,
where $Q$ is a finite set of control states,
$Q_F \subseteq Q$ is the set of final states,
and $\delta$ is a finite set of transitions of the form
\begin{align*}
	&a^r\, (p_1, x_1)\, \dots\, (p_r, x_r) \goesto {} q, t
		\quad \textrm { or } \\%quad
	&p, x_1 \goesto {} q, t \qquad (\textrm{\it $\varepsilon$-transition})
\end{align*}
where $a$ is a letter, $p, q, p_1, \dots, p_r$ are states,
$x_1,\dots,x_r$ are variables of sort $o$,
and $t$ is a term built of variables from $\set{x_1, \dots, x_k}$ ($\set { x_1}$, respectively) and symbols, but no nonterminals.
An FTT $\Aa$ defines in a natural way a binary relation $\trans \Aa$ on trees \cite{tata2007}.
We say that an FTT is \emph{linear} if no term $t$ on the right of transitions contains more than one occurrence of the same variable.

We show that HORSes are closed under linear transductions.
The construction relies on the reflection operation~\cite{broadbent10:_recur_schem_logic_reflec},
in order to detect unproductive subtrees. 
\begin{restatable}{theorem}{thmtransd}
	\label{thm:HORS:transd}
	HORSes are effectively closed under linear tree transductions.
\end{restatable}
A family of word languages is a \emph{full trio} if it is effectively closed under rational (word) transductions.
Since rational transductions on words are a special case of linear tree transductions,
we obtain the following corollary of Theorem~\ref{thm:HORS:transd}.
\begin{corollary}
	\label{cor:HORS:trio}
	Languages of finite words recognized by HORSes form a full trio.
\end{corollary}

\section{The Main Result}
\label{sec:result}

We formulate the main result and state some of its consequences.

\begin{definition}[Diagonal problem]
  For a higher-order recursion scheme $\Ss$, and a set of letters
  $\S$, the predicate $\Unb_\S(\Ss)$ holds if for every $n\in \Nat$
  there is a tree $t$ generated by $S$ with at least $n$ occurrences
  of every letter from $\S$. The \emph{diagonal problem} for schemes is to
  decide whether $\Unb_\S(\Ss)$ holds for a given scheme $\Ss$ and a set $\S$.
\end{definition}

\begin{theorem}\label{thm:main}
  The diagonal problem for higher-order recursion schemes is decidable.
\end{theorem}
\begin{proof}
  The proof is by induction on the order of a HORS $\Ss$. It relies on
  results from the next two sections. 
  If $\Ss$ has order $0$, then $\Ss$ can be converted to an equivalent
  finite automaton on trees,
  for which the diagonal problem can be solved by direct inspection.
%   first
%calculating for every state the alphabets of trees accepted from this
%state, and then analyzing the graph of the automaton decorated with
%this additional information.
%
For $\Ss$ of order greater than $0$, we first convert $\Ss$ to a narrow
HORS $\Ss'$%, and construct sets of symbols $\S'_1,\dots,\S'_k$
such that $\Unb_\S(\Ss)$ holds iff $\Unb_{\S}(\Ss')$ holds %for some $i\in\{1,\dots,k\}$.
(Theorem~\ref{thm:narrowing}).
%This is done in Theorem~\ref{thm:narrowing}.
%
Then, %for every $i\in\set{1,\dots,k}$,
we employ the construction from
Section~\ref{sec:lowering} and obtain a HORS $\Ss''$ of order smaller
by $1$ than the order of $\Ss'$. By Lemmata~\ref{lem:sound}
and~\ref{lem:compl}:~% we have that
$\Unb_{\S}(\Ss')$ holds iff $\Unb_{\S}(\Ss'')$ holds.
% are sets of symbols $\S''_1,\dots,\S''_k$ for which we have:
% $\Unb_{\S'_i}(\Ss')$ holds iff $\Unb_{\S''_i}(\Ss'')$ holds.
\end{proof}

The main theorem allows to solve some other problems for
higher-order schemes. The \emph{downward closure} of a language of words is the set
of its (scattered) subwords. % of the words in the language. 
Since the
subword relation is a well quasi-order~\cite{Higman:1952}, the downward closure of any
language of words is regular. The main theorem implies that the downward
closure can be computed for HORSes generating languages of finite words,
or, in our terminology, $\set{e^0}$-narrow HORSes, where $e^0$ is a
nullary symbol acting as an end-marker.

\begin{corollary}
  There is an algorithm that given an $\set{e^0}$-narrow HORS $\Ss$
  computes a regular expression for the downward closure of the
  language generated by $\Ss$.
\end{corollary}
\begin{proof}
  By Corollary~\ref{cor:HORS:trio},
  word languages generated by schemes are closed under rational transductions.
  In this case,
  Theorem~\ref{thm:main} together with a result of Zetzsche~\cite{Zetzsche:ICALP:2015}
  can be used to compute the downward closure of a language generated by a HORS.
\end{proof}

Piecewise testable languages of words are boolean combinations of languages of
the form $\S^*a_1\S^*a_2\dots\S^*a_k\S^*$ for some
$a_1,\dots,a_k\in\S$.
Such languages talk about possible orders of occurrences of
letters. The problem of separability by piecewise testable languages
asks, for two given languages of words, whether there is a piecewise testable
language of words containing one language and disjoint from the other. A
separating language provides a simple explanation of the disjointness
of the two languages~\cite{hofman_et_al:LIPIcs:2015:4987}.

\begin{corollary}
  There is an algorithm that given two $\set{e^0}$-narrow HORSes
  decides whether there is a piecewise testable language separating the
  languages of the two HORSes.
\end{corollary}
\begin{proof}
  This is an immediate consequence of a result of Czerwi\'nski et
  al.~\cite{CzerwinskiMartensRooijenZeitounZetzsche}
  who show that for any class of languages effectively
  closed under rational transductions, the problem reduces to solving the
  diagonal problem.
\end{proof}
The final example concerns deciding reachability in parameterized
asynchronous shared-memory systems~\cite{DBLP:conf/fsttcs/Hague11}. 
In this model one instance of a process, called leader, communicates with
an undetermined number of instances of another process, called
contributor. 
The communication is implemented by common registers on which the
processes can perform read and write operations; however, operations of
the kind of test-and-set are not possible.
The reachability problem asks if for some number of instances of the
contributor the system has a run writing a designated value to a register. 

\begin{corollary}
  The reachability problem for parameterized asynchronous shared-memory
  systems is decidable for systems where leaders and contributors are
  given by $\set{e^0}$-narrow HORSes.
%, or equivalently higher-order
%  pushdown automata with collapse.
\end{corollary}
\begin{proof}
  La Torre et al.~\cite{LaTorreMuschollWalukiewicz:2015}
  show how to use the downward closure of the
  language of the leader to reduce the reachability problem for a
  parameterized system to the
  reachability problem for the contributor. Being a full trio
  is sufficient for this reduction to work.
\end{proof}

%%% Local Variables:
%%% mode: latex
%%% TeX-master: "main"
%%% End:

%\input{strategy}
\section{Narrowing the HORS}\label{sec:narrowing}

The first step in our proof of Theorem~\ref{thm:main} is to convert a scheme to a narrow
scheme. The property of being narrow is essential for the second step, as lowering
the order of a scheme works only for narrow schemes. This approach
through narrowing has been used by Hague et
al.~\cite{DBLP:conf/popl/HagueKO16} for
higher-order pushdown automata. Here we deal with recursion schemes,
which are equivalent to higher-order pushdown automata with collapse.

The idea behind narrowing is quite intuitive.
Consider a binary tree, and suppose that we are interested in the number
of occurrences of a certain letter $a$, that may appear only in leaves. 
Consider a path that, at each node, selects the subtree containing the larger number of $a$'s,
and let's label the node by $a$ if the successor of the node that is not on the path has an $a$-labeled descendant.
Then, if the original tree had $n$ occurrences of $a$, then on the selected path we put between $\log n$ and $n$ labels $a$.
The lower bound holds since, whenever a subtree is selected,
at most half of the $a$'s is discarded (on the other subtree),
and this happens a number of times equal to the number of $a$'s on the resulting path.
This observation implies it suffices to convert a scheme $\Ss$ generating trees
to a scheme $\Ss'$ generating all paths (words) in the trees generated by $\Ss$ with the additional labeling.
Then  $\Unb_{\set{a}}(\Ss)$ will be equivalent to $\Unb_{\set{a}}(\Ss')$.

The general situation is a bit more complicated
since we are interested in the diagonal problem not just for a single letter,
but for a set of letters $\S$.
In this case, different letters may have different witnessing paths,
so $\Ss'$ should generate not a single path but a narrow tree
whose number of paths is bounded by $|\S|$.

\begin{theorem}
	\label{thm:narrowing}
    For a HORS $\Ss$ and a set of letters $\S$,
	one can construct a set of nullary symbols $\Delta$ of size $|\S|$
	and a $\Delta$-\emph{narrow} HORS $\Ss'$ of the same order as $\Ss$,
	such that $\Unb_\S(\Ss)$ holds if, and only if, $\Unb_{\S}(\Ss')$ holds.
\end{theorem}

\begin{proof}
	We start by assuming that $\Ss$ uses only symbols of rank $2$
        and $0$, where additionally letters from $\Sigma$ appear only
        in leaves.
	The general situation can be easily reduced to this one, by applying a tree transduction that replaces every node by a small fragment of a tree built of binary symbols, with the original label in a leaf.
	
	Then, we consider a linear bottom-up transducer $\Aa$ from trees produced by $\Ss$
	to narrow trees.
	As labels in the resulting trees we use:%
	\begin{inparaenum}[(i)]
	\item	new leaf symbols $\Delta = \set{e_1^0, \dots, e_{|\S|}^0}$,
	\item	unary symbols $a^1$ for all $a\in\Sigma$, and
	\item	new auxiliary symbols $\bullet^k$ (of rank $k\geq 1$).
	\end{inparaenum} 
	For each set of letters $\Gamma \subseteq \Sigma$,
	$\Aa$ contains a state $p_\Gamma^?$ making sure that each letter from $\Gamma$ occurs at least once in the input tree.
	%
	%Moreover, for each natural number $1 \leq n \leq |\Sigma|$, $\Aa$ contains a state $p_n$ which outputs only $n$-narrow trees.
	Moreover, for each nonempty set of leaf labels $\Delta' \subseteq \Delta$,
	$\Aa$ contains a state $p_{\Delta'}$ that outputs only $\Delta'$-narrow trees.
	The final state of $\Aa$  is $p_\Delta$.
	Transitions are as follows:
	\begin{align*}
		\textrm{(Branch)} && a^2 \, (p_{\Delta_1}, x_1) \, (p_{\Delta_2}, x_2)
			&\goesto {} p_{\Delta_1 \cup \Delta_2}, \bullet^2 \, x_1\, x_2\,, \\
		\textrm{(Leaf)} && &\hspace{-3em}a^0
			\goesto {} p_{\set {e_{i_1},\dots,e_{i_k}}}, \bullet^k\,e_{i_1}\,\dots\,e_{i_k}\,, \\
		\textrm{(Choose${}_1$)} && a^2 \, (p_{\Delta_1}, x_1) \, (p_{\Gamma}^?, x_2)
			&\goesto {} p_{\Delta_1}, a_1^1 (\cdots (a_k^1 \, x_1))\,, \\
		\textrm{(Choose${}_2$)} && a^2 \, (p_{\Gamma}^?, x_1) \, (p_{\Delta_2}, x_2)
			&\goesto {} p_{\Delta_2}, a_1^1 (\cdots (a_k^1 \, x_2))\,.
	\end{align*}
	where $\Delta_1$ and $\Delta_2$ are \emph{disjoint} subsets of $\Delta$, where $i_1<\dots<i_k$, 
	and where $ \Gamma = \set{a_1, \dots, a_k}\subseteq\Sigma$.
	%
%	$t_\Delta(\var x)$ is the term $a_1 (a_2 \cdots (a_k \, \var x))$
%	which outputs symbols from $\Delta \cup \set a$ in some order.
	%
	Intuitively, rules of types (Branch) and (Leaf) make sure that we output narrow trees,
	and rules of types (Choose${}_i$) select a branch
	and output (only) letters that appear at least once in the discarded subtree.
	States $p_\Gamma^?$ check that each letter in $\Gamma$ occurs at least once, as follows:
	\begin{align*}
		\textrm{(Check${}_2$)} && a^2 \, (p_{\Gamma_1}^?, x_1) \, (p_{\Gamma_2}^?, x_2)
			&\goesto {} p_{\Gamma_1 \cup \Gamma_2}^?, e_1^0 \\
		\textrm{(Check${}_0$)} && a^0
			&\goesto {} p_{\set a}^?, e_1^0
	\end{align*}
	The set $\trans{p_\Gamma^?}(\set t)$ is either a single leaf or $\emptyset$,
	depending on whether $t$ satisfies the condition or not.
	The choice of $e_1^0$ on the right side of the transitions is not important,
	since, in the way states $p_\Gamma^?$ are used, 
	it only matters whether the input can be successfully parsed,
	and not what the output actually is.

	It is clear that the image of state $p_{\Delta'}$ is always a language of $\Delta'$-narrow trees.
	Correctness follows from the following claim.
	\begin{claim}
		Let $t$ be an input tree.
		Then,%
		\begin{inparaenum}[(i)]
		\item	if $t$ has at least $n$ occurrences of every letter $a \in \Sigma$, then
			$\trans{\Aa}(t)$ contains a tree with at least $\log n$ occurrences of every letter $a \in \Sigma$, and
		\item	if $\trans{\Aa}(t)$ contains a tree with at least $n$ occurrences of every letter $a \in \Sigma$, then
			$t$ has at least $n$ occurrences of every letter $a \in \Sigma$.
		\end{inparaenum}
	\end{claim}
	%
	%We conclude the proof of the theorem by proving the claim.
	%
	\ignore{
	\begin{proof}[Proof of the claim]
		Suppose that every letter $a$ from $\S$ appears at least $n$ times in the tree $t$.
		In this case, by selecting at each step the branch containing the larger number of $a$'s,
		we see that $a$ is output at least $\log n$ times by $\Aa$.
		The other direction is immediate, since every time $\Aa$ outputs a letter $a$,
		an occurrence of this letter appears somewhere in the input tree $t$.
	\end{proof}
	}
	To conclude the proof, let $T$ be the transduction $\trans \Aa$ realized by $\Aa$.
	By Theorem~\ref{thm:HORS:transd},
	there exists a HORS $\Ss'$ of the same order as $\Ss$ with $\lang {\Ss'} = T(\lang \Ss)$.
	First, it is clear that $\lang {\Ss'}$ is a language of $\Delta$-narrow trees.
	Second, thanks to the claim above, $\Unb_\S(\Ss)$ holds if, and only if, $\Unb_\S(\Ss')$ holds.
\end{proof}

%%% Local Variables:
%%% mode: latex
%%% TeX-master: "main"
%%% End:

\section{Lowering the Order}\label{sec:lowering}

Let $\Ss$ be a $\Delta$-narrow HORS of order $k \geq 1$,
and let $\Sigma$ be a finite set of letters.
The goal of this section is to construct a HORS $\Ss'$ of order $k-1$
s.t. $\Unb_\Sigma(\Ss)$ holds if and only if $\Unb_\Sigma(\Ss')$ holds.

Let $\bullet$ be a fresh letter, not used in $\Ss$, and not in $\Sigma$.
We will use it to label auxiliary nodes of trees generated by $\Ss'$.
We say that two trees $K_1$, $K_2$ are \emph{equivalent} if, for each letter $a\neq\bullet$,
they have the same number of occurrences of $a$.
The resulting HORS $\Ss'$ will have the property that for every tree generated by $\Ss$ there exists an equivalent tree generated by $\Ss'$, and for every tree generated by $\Ss'$ there exists an equivalent tree generated by $\Ss$.
Then surely $\Unb_\Sigma(\Ss)$ holds if and only if $\Unb_{\Sigma}(\Ss')$ holds.
%\medskip

Let us explain the idea of lowering the order of a scheme on two simple
examples. Consider the following transformation on sorts that removes arguments of
sort $o$:
\begin{equation*}
o\down =o,\quad \text{and}\quad 
(\beta\to \gamma)\down=
\begin{cases}
\gamma\down & \text{if $\beta=o$,}\\
(\beta\down) \to (\gamma\down) & \text{otherwise.}
\end{cases}
\end{equation*}
We have that the order of $\alpha\down$ is $\max(0,\ord(\alpha)-1)$. 

Very roughly our construction will take a scheme and produce a scheme
of a lower order by changing every nonterminal of sort $\alpha$ to a
nonterminal of sort $\alpha\down$.
This is achieved by outputting immediately arguments of sort $o$ instead of passing them to nonterminals.
%which is reminiscent of \cite{KnapikNiwinskiUrzyczyn:Easy:2002}

\begin{example}\upshape
	Consider the scheme
	\begin{equation*}
	  S\to F\,e^0,\qquad F\,x\to x, \qquad F\,x\to F\,(b^1\,x)\,.
	\end{equation*}
	This scheme generates words of the form $(b^1)^ne^0$. It can be
	transformed to an equivalent scheme:
	\begin{equation*}
	  S'\to\stree{F'}{e^0}
	  \qquad F'\to \bullet^0\qquad F'\to\stree{F'}{b^0}
	\end{equation*}
	where we have used a graphical notation for terms; in standard notation
	the first rule would be $S'\to
	\bullet^2\,F'\,e^0$. 
	Now both $b$ and $e$ are used with rank $0$;
	we have also used auxiliary symbols $\bullet^2$ and $\bullet^0$.
	Observe that the new scheme has smaller order as the sorts of $S'$ and
	$F'$ are $o$.  
	The new scheme is equivalent to the initial one since a derivation of
	$(b^1)^ne^0$ can be matched by the derivation of a tree with one $e^0$ and
	$b^0$ appearing $n$ times:% (see Figure~\ref{fig:ex1}).
	%
	%\begin{figure}[h]
	%  \centering
%	\vspace{-2ex}
	
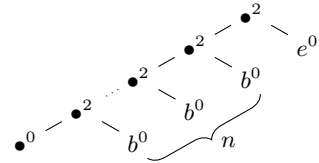
\begin{center}
%	  \hspace{8ex}
	  \begin{tikzpicture}[level distance=3ex]
	    \node{$\bullet^2$} 
	    child{node{$\bullet^2$} 
	      child{node{$\bullet^2$}
	        child[dotted,level distance=3ex]{
	          node{$\bullet^2$}
	          child[solid]{node{$\bullet^0$}} 
	          child[solid]{node(b){$b^0$}}}
	        child{node{$b^0$}}} 
	      child{node(a){$b^0$}}} child{node{$e^0$}};

	     \draw[decorate,decoration={brace,raise=7pt,amplitude=6pt}] 
	  (a.center) --node[midway, auto, outer sep=7pt]{$n$}
	  (b.center);
	  \end{tikzpicture}
	\end{center}
	%  \caption{Trees generated by the new scheme of Example 1}
	%  \label{fig:ex1}
	%\end{figure}
\end{example}
\begin{example}
	\label{ex:two}\upshape
	Let us now look at a more complicated example. This time we take the following
	scheme of order $2$:
	\begin{align*}
	  &S\to F\, b^1\, e^0,\quad\ F\,g\,x\to g\,x,\quad\ F\,g\,x\to a^1\,(F\,(B\,g)\,(c^1\,x)),\\
	 &B\,g\,x\to b^1\,(g\,x)\,.
	\end{align*}
	Here $g$ has sort $o\to o$, and $x$ has sort $o$.
	This scheme generates words of the form $(a^1)^n(b^1)^{n+1}(c^1)^ne^0$. We transform it
	into a scheme of order $1$: 
	\begin{align*}
	  &S'\to\stree{F'\,b^0}{e^0}\quad &&F'\,g'\to g'\\
	&F'\,g'\to \sstree{\node{$\bullet^3$} child{node{$a^0$}} child{node{$F'\,(B'\,g')$}}
	  child{node{$c^0$}};} &&B'\,g'\to\stree{b^0}{g'}
	\end{align*}
	The latter scheme generates trees of the form:
	% given by
	%Figure~\ref{fig:ex2}.

	%\begin{figure}[h]
	%  \centering
	  
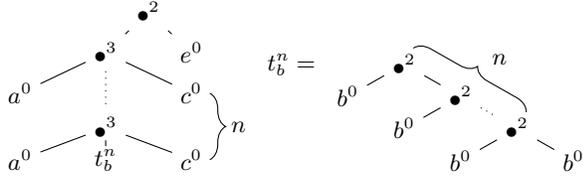
\begin{tikzpicture}[level distance=3.8ex, sibling distance = 8ex]
	    \node{$\bullet^2$}%
	    child{node{$\bullet^3$}%
	      child{node{$a^0$}}%
	      child[dotted,level distance = 7ex]{%
	        node{$\bullet^3$}%
	        child[solid,level distance=3ex]{node{$a^0$}}%
	        child[solid,level distance=3ex]{node{$t_b^{n}$}}%
	        child[solid,level distance=3ex]{node(b){$c^0$}%
	        }}
	    child{node(a){$c^0$}}}%
	  child{node{$e^0$}};%
	  \draw[decorate,decoration={brace,raise=7pt,amplitude=6pt}] 
	  (a.center) --node[midway, auto, outer sep=12pt]{$n$}
	  (b.center);
	  \end{tikzpicture}
	  \begin{tikzpicture}[level distance = 3ex]
	    \node at (-1.5,0){$t_b^{n} = $};

	    \node (a) {$\bullet^2$}%
	    child{node{$b^0$}}%
	    child{node{$\bullet^2$}%
	      child{node{$b^0$}}%
	      child[dotted]{
	        node (b) {$\bullet^2$}%
	        child[solid]{node{$b^0$}}%
	        child[solid]{node{$b^0$}}%       
	      }
	    }%
	    ;%
	      \draw[decorate,decoration={brace,raise=7pt,amplitude=6pt}] 
	  (a.center) --node[midway, auto, outer sep=9pt]{$n$}
	  (b.center);
	  \end{tikzpicture}
	% \caption{Trees generated by the new scheme of Example 2}
	%  \label{fig:ex2}
	%\end{figure}
	%where the first $n+1$ occurrences of $\land$ are followed by the tree
	%$t_b^{n}$ that is a tree with $n+1$ symbols $b'$ arranged 
	%as showed on Figure~\ref{fig:ex2}.
\end{example}

The intuition behind the above two examples is as follows. Consider
some closed term $K$ of sort $o$, and its subterm $L$ of sort $o$. In
a tree generated by $K$, the term $L$ will be used to generate some
subtrees. Take a tree where $L$ generates exactly $k$ subtrees. Then
we can create a new term starting with a symbol $\bullet^{k+1}$:
in the first subtree we put $K$ with $L$ replaced by $\bullet^0$,
and in the $k$ remaining subtrees we put $L$. From this new term
we can generate a tree similar to the initial one: the subtrees
generated by $L$ are moved closer to the root, but the multisets of
letters appearing in the tree do not change. We do this with every
subterm of sort $o$ on the right hand side of every rule of $\Ss$.  In
the obtained system, whenever an argument has sort $o$ then it is
$\bullet^0$. Because of this, we can just drop arguments of sort $o$. This
is what our translation $\alpha\down$ on sorts does, and this is what
happens in the two examples above.
Since the original schemes from the two examples
generated words, and all arguments were eventually used to generate a subword, for
every subterm of sort $o$ the multiplication factor $k$ was always $1$.

The crucial part of this argument was the information on the number of
times $L$ will be used in $K$. This is the main technical problem we
need to address. We propose a special type system for tracking the use
of closures of sort $o$. It will non-deterministically guess the number
of usages, and then enforce derivations that conform to this
guess. The reason why such a \emph{finite} type system can exist is that $\Ss$ is
$\S_0$-narrow, which, in turn, implies that $L$ can be used to generate at most
$|\S_0|$ subtrees of a tree.

\looseness=-1
In the sequel we assume w.l.o.g.~that in $\Ss$ the only rule from the initial nonterminal is $A_\mathit{init}\to A\,e_1^0\,\dots\,e_{|\Delta|}^0$ (for some nonterminal $A$) where $\Delta=\{e_1^0,\dots,e_{|\Delta|}^0\}$, 
and no other rule uses a nullary symbol nor the initial nonterminal $A_\mathit{init}$.
To ensure this condition, we perform the following simple transformation of the HORS.
Every rule $B\,x_1\,\dots\,x_k\to K$ in $\Rr(\Ss)$ is replaced by $B\,y_1\,\dots\,y_{|\Delta|}\,x_1\,\dots\,x_k\to K'$, where $K'$ is obtained by replacing in $K$ every use of a symbol $e_i^0\in\Delta$ by $y_i$, 
and every use of a nullary symbol not being in $\Delta$ by an arbitrary $y_i$ (this symbol anyway does not appear in any tree generated by $\Ss$),
and every use of a nonterminal $C$ by $C\,y_1\,\dots\,y_{|\Delta|}$ (the sort of every nonterminal is changed from $\alpha$ to $o^{|\Delta|}\to\alpha$).
Additionally a new rule $A_\mathit{init}\to A\,e_1^0\,\dots\,e_{|\Delta|}^0$ is added, where $A_\mathit{init}$ is a fresh nonterminal that becomes initial, and $A$ is the nonterminal that was initial previously.
It is easy to see that this transformation does not change the set of
generated trees.
It also does not increase the order,
since in this section we assume that $\Ss$ has order at least $1$.

\begin{figure*}
  \centering
  \begin{mathpar}
    \inferrule{ }{\Gamma,x:\lambda\vdash x:\lambda} \and \inferrule{
    }{\Gamma\vdash A:(\emptyset,\tau)} \and \inferrule{ }{\Gamma\vdash
      a^0:(\{a^0\},\r)} \and \inferrule{r\geq 1}{\Gamma\vdash a^r
      :(\emptyset,\{(S_1,\r)\}\to\dots\to\{(S_r,\r)\}\to\r)} \and
    \inferrule{
      \begin{array}{l}
        \Gamma\vdash L:(S_0,\{(S_1,\tau_1),\dots,(S_k,\tau_k)\}\to\tau)\qquad
        %	\\
        % (S_1,\tau_1)<\dots<(S_k,\tau_k)
%	\\
        \Gamma\vdash M:(S_i,\tau_i)\mbox{ for each }i\in\{1,\dots,k\}
       \end{array}
    }{ \Gamma\vdash L\,M:(S_0\cup S_1\cup\dots\cup S_k,\tau) }
\quad \mbox{ provided that }
        % (S_1,\tau_1)<\dots<(S_k,\tau_k)\mbox{ and }
        S_0\cap(S_1\cup\dots\cup S_k)=\emptyset
  \end{mathpar}
  \caption{Type system for tracing nullary symbols in a term}
  \label{fig:type}
\end{figure*}

\subsection{Type System}

We now present a type system whose main purpose is to track 
nullary symbols that eventually will end as leaves of a generated
tree. The type of a term will say which nullary symbols are already present in
the term and which will come from each of its arguments. 

% We now present a type system, whose main purpose is to assign with each term (and with each argument of a term) a set of nullary symbols that will be used in the tree generated from this term (from this argument).
% When an argument of a term generates a tree with some nullary symbols, then the type system will allow us to control whether this argument is used by a term or not (and how many times it is used).

For every sort $\alpha=(\alpha_1\to\dots\to\alpha_k\to o)$ we define the set $\Tt^\alpha$ of \emph{types} of sort $\alpha$
and the set $\LT^\alpha$ of \emph{labeled types} of sort $\alpha$
by induction on $\alpha$.
Labeled types in $\LT^\alpha$ are just pairs $(S,\tau)\in\Pp(\Delta)\times\Tt^\alpha$, where if $\alpha=o$ we require that $S\neq\emptyset$.
The support of a set $\Lambda$ of labeled types is the subset $\Lambda^{\neq\emptyset}$
of its elements $(S, \tau) \in \Lambda$ with $S \neq \emptyset$.
A set of labeled types $\Lambda$ is \emph{separated}
if there are no two distinct $(S, \tau)$ and $(S', \tau')$ in $\Lambda$
s.t. $S \cap S' \neq \emptyset$.
Types in $\Tt^\alpha$ are of the form $\Lambda_1\to\dots\to\Lambda_k\to\r$,
where $\r$ is a distinguished type corresponding to sort $o$,
$\Lambda_i$ is a subset of $\LT^{\alpha_i}$ for each $i\in\{1,\dots,k\}$
s.t. $\{\Lambda_1^{\neq\emptyset}, \dots, \Lambda_k^{\neq\emptyset}\}$ are pairwise disjoint
and $\Lambda_1 \cup \cdots \cup \Lambda_k$ is separated.
%for every $\lambda_i = (S_i, \tau_i) \in \Lambda_i$ and $\lambda_j = (S_j, \tau_j) \in \Lambda_j$,
%we have $S_i \cap S_j = \emptyset$ whenever $i \neq j$.
% or $i = j$ and $\lambda_i \neq \lambda_j$.
% $(\lambda_i, i) \neq (\lambda_j, j)$.
%
%The last condition simply means that
%we further require that all the sets $S$'s appearing on the first
%coordinate of labeled types in $\Lambda_1, \dots, \Lambda_k$ are pairwise disjoint;
%formally,\igw{I prefer the old formulation}
%
%$|\bigcup_{i\in\{1,\dots,k\}}\bigcup_{(S,\tau)\in\Lambda_i}S|=\sum_{i\in\{1,\dots,k\}}\sum_{(S,\tau)\in\Lambda_i}|S|$.
%for every $\lambda_i = (S_i, \tau_i) \in \Lambda_i$ and $\lambda_j = (S_j, \tau_j) \in \Lambda_j$,
%we have $S_i \cap S_j = \emptyset$ whenever either $i \neq j$ or $i = j$ and $\lambda_i \neq \lambda_j$.
% $(\lambda_i, i) \neq (\lambda_j, j)$.
%
Let us emphasize that $\Lambda_i$ for $\alpha_i=o$ can only contain pairs $(S,\tau)$ with $S\neq\emptyset$.
We fix some (arbitrary) order $<$ on elements of $\LT^\alpha$ for every sort $\alpha$.

%For every sort $\alpha=(\alpha_1\to\dots\to\alpha_k\to o)$ we define the set $\Tt^\alpha$ of \emph{types} of sort $\alpha$, by induction on the structure of $\alpha$.
%Types in $\Tt^\alpha$ are of the form $\Lambda_1\to\dots\to\Lambda_k\to\r$, where $\Lambda_i\subseteq\Pp(\Delta)\times\Tt^{\alpha_i}$ for each $i\in\{1,\dots,k\}$, and
%\begin{itemize}
%	\item	$|\bigcup_{i\in\{1,\dots,k\}}\bigcup_{(S,\tau)\in\Lambda_i}S|=\sum_{i\in\{1,\dots,k\}}\sum_{(S,\tau)\in\Lambda_i}|S|$, and
%%	\item	for each $i\in\{1,\dots,k\}$ such that $\alpha_i\neq o$ it holds that $\{\emptyset\}\times\Tt^{\alpha_i}\subseteq\Lambda_i$, and
%	\item	for each $i\in\{1,\dots,k\}$ such that $\alpha_i=o$ it holds that $(\emptyset,\r)\not\in\Lambda_i$. %$(\\{\emptyset\}\times\Tt^{\alpha_i})\cap\Lambda_i=\emptyset$.
%\end{itemize}
%Let us comment on these two conditions.
%The first of them means that the sets appearing in the pairs $(S,\tau)\in\Lambda_i$ are all disjoint.
%The second condition means that $\Lambda_i$ for $\alpha_i=o$ can only contain pairs $(S,\tau)$ with $S\neq\emptyset$, since the only element of $\Tt^o$ is $\r$.
%We also denote $\LT^\alpha=\Pp(\Delta)\times\Tt^\alpha$; elements of $\LT^\alpha$ are called \emph{labeled types}.
%We fix some (arbitrary) order on elements of $\LT^\alpha$ for every sort $\alpha$.

Types do not describe all the possible trees generated by a term,
but rather restrict the generating power of a term. % to some particular cases.
Intuitively, a labeled type $(S_0,\r)$ assigned to a closed term of sort $o$ says
that we are interested in generating trees that are $S_0$-narrow.
%moreover it guarantees that at least one $S$-narrow tree can be
%generated.
A functional type $(S_0,\Lambda\to\tau)$ says that the term becomes of type
$(S,\tau)$ when taking an argument that will be used only with labeled types from $\Lambda$.
Here, $S$ equals $S_0$ plus the symbols $S_1 \cup \cdots \cup S_k$
generated by an argument of type $\Lambda = \set{(S_1,\tau_1),\dots,(S_k,\tau_k)}$.
%

% Intuitively, a term has type $\Lambda\to\tau$ when it can return $\tau$, while taking an argument for which we can derive all labeled types from $\Lambda$.
% The labeled type extends this type by a set $S\subseteq\Delta$ of nullary symbols used in the term, including (terms that will be substituted for) its free variables.

A \emph{type environment} $\Gamma$ is a set of bindings of variables of the form $x^\alpha:\lambda$, where $\lambda\in\LT^\alpha$;
we may have multiple bindings $x^\alpha:\lambda_1, \dots, x^\alpha:\lambda_n$ for the same variable
(which we also abbreviate as $x^\alpha:\set{\lambda_1, \dots, \lambda_n}$),
however $\set{\lambda_1, \dots, \lambda_n}$ must be separated in the sense above.
A \emph{type judgment} is of the form $\Gamma\vdash
M^\alpha:\lambda$, where again $\lambda\in\LT^\alpha$. % for $\alpha$
                                % being the sort of $M$.

The rules of the type system are given in Figure~\ref{fig:type}.
A \emph{derivation} is a tree whose nodes are labeled by
type judgments constructed according to the rules of the type system  % of Figure~\ref{fig:type}
(we draw a parent below its children,
unlikely the usual convention for trees).
For the proof it will be convenient to assume that a derivation is an ordered tree:
in the application rule the premise with $L$ is the first sibling
followed by the premises with $M$ ordered using our fixed ordering on $(S_i,\tau_i)$, without repetitions.
We say that $D$ is a \emph{derivation for} $\Gamma\vdash M:\lambda$, or that $D$ \emph{derives} $\Gamma\vdash M:\lambda$, if this type judgment labels the root of $D$.
All the nodes of derivations are required to be labeled by valid type judgments,
thus all the restrictions on types from the definition of $\Tt^\alpha$ stay in force;
in particular, in the application rule for $L M$, the sets $S_1, \dots, S_k$ are disjoint.

\subsection{Transformation}

Once we have the type system, we can show how the HORS $\Ss$ is transformed into the HORS $\Ss'$.

A term of type $\tau$ will be transformed into a term of sort $\tr(\tau)$.
This sort is defined by induction on the structure of $\tau$, as follows:
\begin{itemize}
	\item	$\tr(\r)=o$, and
	\item	if $\tau=(\Lambda\to\tau')\in\Tt^{\alpha\to\beta}$ with $\Lambda \! = \! \set{(S_1,\tau_1)\!<\!\dots\!<\!(S_k,\tau_k)}$,
	then	we have
		\begin{align*}
			&\tr(\tau)=\left\{\begin{array}{ll}
				\tr(\tau_1)\to\dots\to\tr(\tau_k)\to\tr(\tau')&\mbox{if }\alpha\neq o,\\
				\tr(\tau')&\mbox{if }\alpha=o.
			\end{array}\right.
		\end{align*}
\end{itemize}
We see that if $\tau\in\Tt^\alpha$, then
$\ord(\tr(\tau))=\max(0,\ord(\alpha)-1)$.
This translation is a refined version of the translation
$\alpha\down$ on sorts that we have seen earlier in the examples.

The nonterminals of $\Ss'$ will be the nonterminals of $\Ss$ labeled
with types.
For every nonterminal $A$ from $\Ss$, of some sort $\alpha$, and for every $\tau$ such that $(\emptyset,\tau)\in\LT^\alpha$, in $\Ss'$ we consider a nonterminal $A\restr_\tau$ of sort $\tr(\tau)$.
Moreover, for every variable $x$ used in $\Ss$, being of some sort $\alpha\neq o$, and for every $\lambda=(S,\tau)\in\LT^\alpha$, in $\Ss'$ we consider a variable $x\restr_\lambda$ of sort $\tr(\tau)$.

Before defining the rules of $\Ss'$, we need to explain how to
transform terms to match the transformation on types.
This transformation is guided by derivations.
We define a term $\tr(D)$, where $D$ is a derivation for $\Gamma\vdash K:\lambda$, as follows:
\begin{itemize}
	\item	If $K=a^r$ is a symbol, then $\tr(D)=a^0$.
	\item	If $K=x^\alpha$ is a variable, then $\tr(D)=\bullet^0$ if $\alpha=o$, and $\tr(D)=x\restr_\lambda$ otherwise.
	\item	If  $K=A$ is a nonterminal, then $\tr(D)=A\restr_\tau$
          provided that $\lambda=(\emptyset,\tau)$.
	\item	Suppose that $K=L\,M$ is an application.
		Then in $D$ we have a subtree $D_0$ deriving $\Gamma\vdash L:(S_0,\Lambda\to\tau)$,
		where $\Lambda = \set {\lambda_1<\dots<\lambda_k}$,
		and for each $i\in\{1,\dots,k\}$ a subtree $D_i$ deriving $\Gamma\vdash M:\lambda_i$.
		If the sort of $M$ is $o$, then we take $\tr(D)=\tr(D_0)$;
		otherwise, $\tr(D)=\tr(D_0)\,\tr(D_1)\,\dots\,\tr(D_k)$.
\end{itemize}
We notice that for $\lambda=(S,\tau)$ the sort of $\tr(D)$ is indeed $\tr(\tau)$.

We see that arguments of sort $o$ are ignored while transforming an application.
Because of that, we need to collect the result of the transformation for all those subtrees of the derivation that describe terms of sort $o$.
This is realized by the $\trcum$ operation that returns a list of terms of sort $o$.
When $D$ is a derivation for a term of sort $\alpha$, and subtrees of $D$ starting in the children of the root are $D_1,\dots,D_m$, then
\begin{align*}
	\trcum(D)=\left\{\!\!\!\begin{array}{ll}
		\tr(D);\trcum(D_1);\dots;\trcum(D_m)&\mbox{if }\alpha=o,\\
		\trcum(D_1);\dots;\trcum(D_m)&\mbox{otherwise.}\!
	\end{array}\right.
\end{align*}
For a list $R_1;\dots;R_k$ of terms of sort $o$, let us define the term
$\merge(R_1;\dots;R_k)$ as $\bullet^k\,R_1\,\dots\,R_k$.
Finally, for a substitution $\eta$, and a list of terms $\lista$ we
write $\lista[\eta]$ for the list where the substitution is performed on
every term of the $\lista$.

\begin{example}\upshape
	To see an example of such a translation take the term $b^1\,(g\,x)$ that is
	on the right side of the rule for $B$ in Example~\ref{ex:two}.
	For readability of types, we write $(e,\r)$ instead
	of $(\set{e^0},\r)$. We take an environment  $\Gamma\equiv x:(e,\r),
	g:(\es,\set{(e,\r)}\to \r)$ and consider the derivation presented in
	Figure~\ref{fig:der}.
	\begin{figure}[tb]
	\resizebox{7.7cm}{!}{
	\begin{mathpar}
	  \inferrule{
	{\Gamma\vdash b^1:(\emptyset,\set{(e,\r)}\to \r)}\quad 
	\inferrule{\Gamma\vdash g:(\es,\set{(e,\r)}\to \r)\quad \Gamma\vdash x:(e,\r)}{\Gamma\vdash g\,x : (e,\r)}}
	{\Gamma\vdash b^1\,(g\,x):(e,\r)}
	\end{mathpar}  
	}
	  \caption{An example derivation}
	  \label{fig:der}
	\end{figure}
	Calling this derivation $D$, we have $\tr(D)=b^0$ and
	$\trcum(D)=b^0;g\restr_{(\es,\set{(e,\r)}\to 
	  \r)};\bullet^0$. Finally, 
	%\begin{align*}
		$\merge(\trcum(D))=\bullet^3\,b^0\,{g\restr_{(\es,\set{(e,\r)}\to\r)}}\,\bullet^0$
	%\end{align*}
	is the (slightly perturbed) result of the transformation of the rule for $B$ in Example~2.
\end{example}

The new HORS $\Ss'$ is created as follows.
The rule $A_\mathit{init}\to A\,e_1^0\,\dots\,e_{|\Delta|}^0$ from the initial nonterminal of $\Ss$ 
is replaced by $A_\mathit{init}\to\merge(A_{\tau_0};e_1^0;\dots;e_{|\Delta|}^0)$ where $\tau_0=\{(\{e_1^0\},\r)\}\to\dots\to\{(\{e_{|\Delta|}^0\},\r)\}\to\r$.
For every other rule of $\Ss$ of the form
$A^\alpha\,x_1^{\alpha_1}\,\dots\,x_k^{\alpha_k}\to K$ we create a
rule in $\Ss'$ for every derivation of $K$. 
More precisely, for each $i\in\{1,\dots,k\}$
consider the (separated) set of labeled types $\Lambda_i = \set{\lambda_{i,1}<\dots<\lambda_{i,n_i}}$,
where $\lambda_{i,j}=(S_{i,j},\tau_{i,j})$ for every $i,j$.
For every derivation $D$ of the form
$x_1 : \Lambda_1, \dots, x_k : \Lambda_k \vdash
K:(\bigcup_{i\in\{1,\dots,k\}}\bigcup_{j\in\{1,\dots,n_i\}}S_{i,j},\r)$
we create a rule 
\begin{align*}
	A_\tau\,\Vars{1}\,\dots\,\Vars{k}\to\merge(\trcum(D))\,,
\end{align*}
where $\tau=(\Lambda_1\to\dots\to\Lambda_k\to\r)\in\Tt^\alpha$,
and $\Vars{i}$ denotes $x_i\restr_{\lambda_{i,1}}\,\dots\,x_i\restr_{\lambda_{i,n_i}}$ if $\alpha_i\neq o$,
and the empty sequence of variables if $\alpha_i=o$ (for $i\in\{1,\dots,k\}$).

The correctness of the transformation is described by the following two lemmata,
which are proved in the next two subsections. 
Their statements refer to the notion of equivalence introduced at the beginning of this section.

\begin{restatable}[Soundness]{lemma}{lemmasoundness}
  \label{lem:sound}
	For every tree generated by $\Ss'$ there exists an equivalent tree generated by $\Ss$.
\end{restatable}

\begin{restatable}[Completeness]{lemma}{lemmacompleteness}
  \label{lem:compl}
	For every tree generated by $\Ss$ there exists an equivalent tree generated by $\Ss'$.
\end{restatable}

\subsection{Soundness}

To prove Lemma~\ref{lem:sound},
we follow a sequence of reductions of $\Ss'$,
and we construct corresponding reductions of $\Ss$.
We however need to assume that the sequence of reductions in $\Ss'$ is leftmost.
We write $P\to_{\Ss'}^\mathit{lf}P'$ to denote that this is the \emph{leftmost} reduction: in $\bullet^k\,P_1\,\dots\,P_k$ we can reduce inside $P_i$ only when in $P_1,\dots,P_{i-1}$ there are no more nonterminals.
Not surprisingly, the order of reductions does not influence the final result,
as stated in the following lemma.

\begin{restatable}{lemma}{lemmasleftmost}
  \label{lem:s-leftmost}
	Suppose that a tree $Q$ can be reached from a term $P$ using some sequence of reductions of $\Ss'$.
	Then $Q$ can be reached from $P$ using a sequence of reductions of $\Ss'$ of the same length in which all reductions are leftmost. %that additionally starts with a leftmost reduction.
\end{restatable}

%\begin{proof}
%	Obvious: when $P=\wedge\,P_1\,P_2$, it does not matter whether we first perform reductions in $P_1$ or first in $P_2$.
%
%	HERE IS A MORE PRECISE PROOF, BUT I THINK THAT IT IS REDUNDANT
%
%	Induction on the size of $P$.
%	In the base case, when the head of $P$ is a nonterminal, we all reductions that can be performed from $P$ are leftmost.
%	Suppose that $P=\wedge\,P_1\,P_2$.
%	Then surely $Q=\wedge\,Q_1\,Q_2$; the sequence of reductions $P\to_\Ss'^* Q$ can be split into reductions that are applied to $P_1$ and those that are applied to $P_2$,
%	and thus into sequences $P_1\to_\Ss'^* Q_1$ and $P_2\to_\Ss'^* Q_2$.
%	By the induction assumption we can assume that the sequence of reductions $P_1\to_\Ss'^* Q_1$ starts with a leftmost reduction (without changing the number of reductions).
%	Equally well we can perform first all reductions in the $P_1$ part, and after that all reductions in the $P_2$ part.
%	In this way we obtain a sequence of reductions that starts with a leftmost reduction and is of the same length as the original one.
%\end{proof}

We need to generalize the definition of equivalence from trees to (lists of) terms of sort $o$ possibly containing nonterminals.
We say that two lists of terms of sort $o$ are \emph{$\merge$-equivalent} if one can be obtained from the other by:
\begin{itemize}
\item	permuting its elements,
\item	adding or removing the $\bullet^0$ term, 
\item	merging/unmerging some list elements using the symbol $\bullet^k$.
\end{itemize}
%We say that $\lista_1$ reduces in $\Ss'$ to $\lista_2$, denoted $\lista_1\to_{\Ss'}^\approx\lista_2$, 
%if there are lists $\lista_1'$ and $\lista_2'$ $\merge$-equivalent to $\lista_1$ and $\lista_2$, respectively, 
%and such that the first element of $\lista_1'$ reduces in $\Ss'$ to the first element of $\lista_2'$ and all other elements of these lists remain the same.
%We notice that if $\lista_1$ and $\lista_2$ are $\merge$-equivalent (so also when $\lista_1\to_{\Ss'}^\approx\lista_2$), 
%then for every tree generated from $\merge(\lista_2)$ there exists an equivalent tree generated from $\merge(\lista_1)$.
The following property of \emph{$\merge$-equivalent} lists should be clear.

\begin{restatable}{lemma}{lemmasmergeequiv}
	\label{lem:s-merge-equiv}
	Let $\lista$ and $\lista'$ be two $\merge$-equivalent lists of terms of sort $o$.
	Suppose that a tree $Q$ can be generated by $\Ss'$ from $\merge(\lista)$.
	Then some tree $Q'$ equivalent to $Q$ can be generated by $\Ss'$ from $\merge(\lista')$ using a sequence of reductions of the same length.
\end{restatable}

The next lemma contains an important observation, needed later in the proof of Lemma \ref{lem:s-subst}.

\begin{lemma}\label{lem:s-empty}
	If $D$ derives $\vdash K:(\emptyset,\tau)$, then $\trcum(D)$ is empty.
\end{lemma}

\begin{proof}
	By induction on the structure of $D$.
	Recall that $\LT^o$ does not contain pairs with $\emptyset$ on the first coordinate, so the sort of $K$ is not $o$, 
	and thus $\trcum(D)$ is defined as the concatenation of $\trcum(\cdot)$ for the subtrees of $D$ starting in the children of the root.
	When $D$ consists of a single node, we immediately have that $\trcum(D)$ is empty.
	Otherwise $K=L\,M$, and the subtrees of $D$ starting in the children of the root are $D_0$ deriving $\vdash L:(S_0,\{(S_1,\tau_1),\dots,(S_k,\tau_k)\}\to\tau)$ and $D_i$ deriving $\vdash M:(S_i,\tau_i)$ for $i\in\{1,\dots,k\}$.
	Since $\emptyset=S_0\cup\dots\cup S_k$, we have $S_i=\emptyset$ for every $i\in\{0,\dots,k\}$.
	The induction assumption implies that $\trcum(D_i)$ is empty for every $i\in\{0,\dots,k\}$, and thus $\trcum(D)$ is empty.
\end{proof}

Before relating reductions in the HORSes, we analyze what happens during a substitution.

\begin{lemma}\label{lem:s-subst}
Consider a derivation $D_K$ for $\Gamma\vdash K^{\alpha_K}:(S,\tau)$.
Suppose all bindings for a variable $x^{\alpha_x}$ in $\Gamma$ are
$(x:\lambda_1),\dots,(x:\lambda_k)$, where
$\lambda_1<\dots<\lambda_k$, and
$(\{\lambda_1,\dots,\lambda_k\}\to\r)$ is a type in $\Tt^{\alpha_x\to o}$.
Suppose also that we have a closed term $N^{\alpha_x}$ with, for every
$i=1,\dots,k$, a  derivation $D_i$ for $\vdash N:\lambda_i$.
Then there is a derivation $D'$ for $\Gamma\vdash K[N/x]:(S,\tau)$
such that $\trcum(D')$ is $\merge$-equivalent to
$\trcum(D_K)[\eta];\trcum(D_{i_1});\dots;\trcum(D_{i_m})$
where
$\eta=(\tr(D_1)/x\restr_{\lambda_1},\dots,\tr(D_k)/x\restr_{\lambda_k})$,
and $i_1<\dots<i_m$ are  those among $i\in\{1,\dots,k\}$ for which
in $D_K$ there is a node labeled by $\Gamma\vdash x:\lambda_i$.
Moreover, if $\alpha_K\neq o$ then $\tr(D')=\tr(D_K)[\eta]$.
\end{lemma}

\begin{proof}
	Induction on the structure of $K$.
	%Let $\alpha_K$ be the sort of $K$.
	We consider three cases.
	
	The trivial case is when $K$ is a nonterminal, or a symbol, or a variable other than $x$.
	Then $K[N/x]=K$, so as $D'$ we can take $D_K$.
	Notice that we have $m=0$ and that the substitution $\eta$ does
        not change neither $\trcum(D_K)$ nor $\tr(D_K)$ since
        variables $x\restr_{\lambda_i}$ do not appear in these
        terms. %So the thesis holds trivially.
	
	Another easy case is when $K=x$, and thus $(S,\tau)=\lambda_l$ for some $l\in\{1,\dots,k\}$.
	We have $m=1$, and $i_1=l$, and $K[N/x]=N$.
	The required derivation $D'$ is obtained from $D_l$ by prepending the type judgment in every its node by the type environment $\Gamma$.
	Clearly $D'$ remains a valid derivation, and $\tr(D')=\tr(D_l)$ and $\trcum(D')=\trcum(D_l)$.
	We see that $\trcum(D_K)$ is either the empty list (when $\alpha_K\neq o$) or $\bullet^0$ (when $\alpha_K=o$), so attaching $\trcum(D_K)[\eta]$ does not change the class of $\merge$-equivalence.
	If $\alpha_K\neq o$, we have $\tr(D_K)[\eta]=x\restr_{\lambda_l}[\eta]=\tr(D_l)$.  %%%% This is not true for $\alpha_K=o$, since then $\tr(D_K)=\top$ !!!!
	
	A more involved case is when $K=L^{\alpha_L}\,M^{\alpha_M}$.
	Then in $D_K$, below its root, we have a subtree $C_0$ deriving $\Gamma\vdash L:(S_0,\{(S_1,\tau_1),\dots,(S_n,\tau_n)\}\to\tau)$, and for each $j\in\{1,\dots,n\}$ a subtree $C_j$ deriving $\Gamma\vdash M:(S_j,\tau_j)$,
	where $(S_1,\tau_1)<\dots<(S_n,\tau_n)$, and $S_0\cap(S_1\cup\dots\cup S_n)=\emptyset$, and $S=S_0\cup\dots\cup S_n$.
	We apply the induction assumption to all these subtrees, obtaining a derivation $C_0'$ for $\Gamma\vdash L[N/x]:(S_0,\{(S_1,\tau_1),\dots,(S_n,\tau_n)\}\to\tau)$ and for each $j\in\{1,\dots,n\}$ a derivation $C_j'$ for $\Gamma\vdash M[N/x]:(S_j,\tau_j)$.
	We compose these derivations into a single derivation $D'$ for $\Gamma\vdash K[N/x]:(S,\tau)$ using the application rule.
	It remains to prove the required equalities about $\trcum$ and $\tr$.

	Let us first see that $\tr(D')=\tr(D_K)[\eta]$ (not only if $\alpha_K\neq o$, but also if $\alpha_K=o$).
	From the induction assumption we know that $\tr(C_0')=\tr(C_0)[\eta]$, as surely $\alpha_L\neq o$. %the sort of $L$ is not $o$.
	%Let $\alpha_M$ be the sort of $M$.
	If $\alpha_M=o$, we simply have $\tr(D')=\tr(C_0')$ and $\tr(D_K)=\tr(C_0)$, so clearly $\tr(D')=\tr(D_K)[\eta]$ holds.
	If $\alpha_M\neq o$, from the induction assumption we also know that $\tr(C'_j)=\tr(C_j)[\eta]$ for every $j\in\{1,\dots,n\}$;
	we have $\tr(D')=\tr(C_0')\,\tr(C'_1)\,\dots\,\tr(C'_n)$ and similarly $\tr(D_K)=\tr(C_0)\,\tr(C_1)\,\dots\,\tr(C_n)$, so we also obtain $\tr(D')=\tr(D_K)[\eta]$.
	
	Next, we prove that $\trcum(D')$ is $\merge$-equivalent to the
        list $\trcum(D_K)[\eta];\trcum(D_{i_1});\dots;\trcum(D_{i_m})$.
	For each $j\in\{0,\dots,n\}$, let $i_{j,1}<\dots<i_{j,m_j}$ be those among $i\in\{1,\dots,k\}$ for which in $C_j$ there is a node labeled by $\Gamma\vdash x:\lambda_i$.
	By definition $\trcum(D')$ consists of $\trcum(C_j')$ for $j\in\{0,\dots,n\}$, and if $\alpha_K=o$ then also of $\tr(D')$.
	Similarly, $\trcum(D_K)[\eta]$ consists of
        $\trcum(C_0)[\eta];\dots;\trcum(C_n)[\eta]$, and
        of $\tr(D_K)[\eta]$ if $\alpha_K=o$. 
	We have already shown that $\tr(D')=\tr(D_K)[\eta]$.
	The induction assumption implies that $\trcum(C_j')$ is $\merge$-equivalent to $\trcum(C_j)[\eta];\trcum(D_{i_{j,1}});\dots;\trcum(D_{i_{j,m_j}})$ for each $j\in\{0,\dots,n\}$.
	It remains to observe that the concatenation of the lists $\trcum(D_{i_{j,1}});\dots;\trcum(D_{i_{j,m_j}})$ for $j\in\{0,\dots,n\}$ is $\merge$-equivalent to $\trcum(D_{i_1});\dots;\trcum(D_{i_m})$.
	By definition every $i_{j,l}$ equals to some $i_{l'}$ and every $i_{l}$ equals to some $i_{j,l'}$; the only question is about duplicates on these lists.
	Let us write $\lambda_i=(T_i,\sigma_i)$ for every $i\in\{1,\dots,k\}$.
	When some $i_l$ is such that $T_{i_l}=\emptyset$, then the list $\trcum(D_{i_l})$ is empty (Lemma \ref{lem:s-empty}),
	so anyway we do not have to care about duplicates.
	On the other hand, when $T_{i_l}\neq\emptyset$ and a node labeled by $\Gamma\vdash x:\lambda_{i_l}$ appears in some $C_j$, then $T_{i_l}\subseteq S_j$.
	Since the sets $S_0,\dots,S_n$ are disjoint, such node appears in $C_j$ only for one $j$, and thus such $i_l$ equals to only one among the $i_{j,l'}$'s.
\end{proof}

We can now formulate and prove the key lemma of this section, allowing us to simulate a single step of $\Ss'$ by a single step of $\Ss$.

\begin{lemma}\label{lem:s-step}
	Let $D$ be a derivation for $\vdash L:(S,\r)$, where $L$ does not contain the initial nonterminal of $\Ss$.
	If $\merge(\trcum(D))\to^\mathit{lf}_{\Ss'}P$, then there
        exists a term $L'$ and a derivation $D'$ for $\vdash L':(S,\r)$ such that $L\to_\Ss L'$ and $\trcum(D')$ is $\merge$-equivalent to $P$.
\end{lemma}

\begin{proof}
	We proceed by induction on the structure of $L$.

	Suppose first that $L=a^r\,M_1\,\dots\,M_r$ (where surely $r\geq 1$).
	Then $D$ starts with a sequence of $r$ application rules followed
        by a single-node derivation for $\vdash
        a^r:(\emptyset,\{(S_1,\r)\}\to\dots\to\{(S_r,\r)\}\to\r)$, and
        by derivations $D_i$ for $\vdash M_i:(S_i,\r)$, for each
        $i\in\{1,\dots,r\}$. In particular, $S_1,\dots,S_r$ are
        disjoint and their union is $S$. 
	It holds that $\trcum(D)=(a^0;\trcum(D_1);\dots;\trcum(D_r))$.
	The reduction $\merge(\trcum(D))\to^\mathit{lf}_{\Ss'}P$ concerns one of terms on one of the lists $\trcum(D_i)$,
	and thus we can write $P=\merge(a^0;\lista_1';\dots;\lista_r')$, where for some $l\in\{1,\dots,r\}$ we have $\merge(\trcum(D_l))\to_{\Ss'}^\mathit{lf}\merge(\lista_l')$, and $\trcum(D_i)=\lista_i'$ for $i\neq l$.
	We apply the induction assumption to $M_l$, obtaining a term $M_l'$ and a derivation $D_l'$ for $\vdash M_l':(S_l,\r)$ such that $M_l\to_\Ss M_l'$ and that $\trcum(D_l')$ is $\merge$-equivalent to $\lista_l'$.
	Taking $D_i'=D_i$ and $M_i'=M_i$ for $i\neq l$, and $L'=a^r\,M_1'\,\dots\,M_r'$, we have $L\to_\Ss L'$.
	Out of a node labeled by $\vdash a^r:(\emptyset,\{(S_1,\r)\}\to\dots\to\{(S_r,\r)\}\to\r)$ and of derivations $D_i'$ for $i\in\{1,\dots,r\}$ we compose a derivation $D'$, using the application rule $r$ times.
	We have $\trcum(D')=(a^0;\trcum(D'_1);\dots;\trcum(D'_r))$, and thus $\trcum(D')$ is $\merge$-equivalent to $P$.
	
	The remaining possibility is that $L=A\,N_1^{\alpha_1}\,\dots\,N_k^{\alpha_k}$.
	Then $D$ starts with a sequence of application rules ending in
        a single-node derivation for $\vdash A:(\emptyset,\tau)$ with $\tau=\{\lambda_{1,1},\dots,\lambda_{1,n_1}\}\to\dots,\to\{\lambda_{k,1},\dots,\lambda_{k,n_k}\}\to\r$,
	and in derivations $D_{i,j}$ for $\vdash N_i:\lambda_{i,j}$, for
        each $i\in\{1,\dots,k\}$, $j\in\{1,\dots,n_i\}$. 
        Suppose $\lambda_{i,1}<\dots<\lambda_{i,n_i}$ for every
        $i\in\{1,\dots,k\}$, and
        $\lambda_{i,j}=(S_{i,j},\tau_{i,j})$ for every $i,j$. 
	Since we consider the leftmost reduction of $\merge(\trcum(D))$, it necessarily concerns its part $\tr(D)$ (which is the first term in the list $\trcum(D)$), 
	that consists of the nonterminal $A\restr_\tau$ to which some of the terms $\tr(D_{i,j})$ are applied (namely, terms $\tr(D_{i,j})$ for those $i$ for which $\alpha_i\neq o$).
	This reduction uses some rule $A_\tau\,\Vars{1}\,\dots\,\Vars{k}\to\merge(\trcum(D_K))$, 
	where in $\Ss$ we have a rule $A\,x_1\,\dots\,x_k\to K$, and we have a derivation $D_K$ for $\Gamma\vdash K:(S,\tau)$ with $\Gamma=\bigcup_{i\in\{1,\dots,k\}}\bigcup_{j\in\{1,\dots,n_i\}}\{x_i:\lambda_{i,j}\}$,
	and where $\Vars{i}$ denotes $x_i\restr_{\lambda_{i,1}}\,\dots\,x_i\restr_{\lambda_{i,n_i}}$ if $\alpha_i\neq o$ and the empty sequence of variables if $\alpha_i=o$ (for $i\in\{1,\dots,k\}$). 

	As $L'$ we take the result of applying the rule
$A\,x_1\,\dots\,x_k\to K$ to $L$,
i.e.~$L'=K[N_1/x_1,\dots,N_k/x_k]$. To construct a derivation for it,
we construct derivations $D_{i,K,}$ for $K[N_1/x_1,\dots,N/x_i]$, for
$i=1,\dots,k$. We take $D_{0,K}=D_K$. To obtain $D_{i,K}$ we apply
Lemma \ref{lem:s-subst} to $N_i$, $D_{i-1,K}$ and $D_{i,1},\dots,D_{i,n_i}$.
	The derivation $D_{k,K}$ derives $\Gamma\vdash L':(S,\r)$.
	Let $D'$ be the derivation for $\vdash L':(S,\r)$ obtained from $D_{k,K}$ by removing the type environment $\Gamma$ from type judgments in all its nodes;
	we obtain a valid derivation since $L'$ is closed.

	It remains to see that $\trcum(D')$ is $\merge$-equivalent to $P$.
	Let $\lista$ be the concatenation of lists $\trcum(D_{i,j})$ for all $i\in\{1,\dots,k\}$, $j\in\{1,\dots,n_i\}$
	and let $Q=\merge(\trcum(D_K))[\eta_1,\dots,\eta_k]$ where $\eta_i=(\tr(D_{i,1})/x_i\restr_{\lambda_{i,1}},\dots,\tr(D_{i,n_i})/x_i\restr_{\lambda_{i,n_i}})$ for $i\in\{1,\dots,k\}$;
	we see that $Q$ is the result of applying the considered rule to $\tr(D)$ (substitutions $\eta_i$ for $i$ such that $\alpha_i=o$ can be skipped, since anyway variables $x_i\restr_{\lambda_{i,j}}$ for such $i$ do not appear in $\tr(D_K)$).
	For $i\in\{1,\dots,k\}$, let $j_{i,1}<\dots<j_{i,m_i}$ be those among $j\in\{1,\dots,n_i\}$ for which in $D_K$ there is a node labeled by $\Gamma\vdash x:\lambda_{i,j}$.
	By definition $\trcum(D)=(\tr(D);\lista)$, and thus $P=\merge(Q;\lista)$.
	On the other hand Lemma \ref{lem:s-subst} says that $\trcum(D')$ is $\merge$-equivalent to $Q;\lista'$, where $\lista'$ is the concatenation of $\trcum(D_{i,j_1});\dots;\trcum(D_{i,j_{m_i}})$ for $i\in\{1,\dots,k\}$.
	We notice, however, that $\lista=\lista'$.
	Indeed, if $S_{i,j}=\emptyset$ for some $i,j$, then $\trcum(D_{i,j})$ is empty by Lemma \ref{lem:s-empty}.
	Suppose that $S_{i,j}\neq\emptyset$.
	The rules of the type system ensure that the subset of $\Delta$ in the root of $D_K$ (that is $S$) is the union of those subsets in all leaves of $D_K$.
	We have assumed that symbols from $\Delta$ do not appear in $K$ (they are allowed to appear only in the rule from the initial nonterminal).
	Moreover, $S_{i,j}\subseteq S$, and all other sets $S_{i'\!,j'}$ are disjoint from $S_{i,j}$ (by the definition of types).
	Thus necessarily a node labeled by $\Gamma\vdash x_i:\lambda_{i,j}$ appears in $D_K$ (this is the only way the elements of $S_{i,j}$ can be introduced in $S$).
	This means that our $j$ is listed among $j_{i,1},\dots,j_{i,m_i}$, and hence $\trcum(D_{i,j})$ appears in $\lista'$.
	This proves that $\lista=\lista'$, and in consequence that $\trcum(D')$ is $\merge$-equivalent to $P$.
\end{proof}

\begin{lemma}\label{lem:s-end}
	Let $D$ be a derivation for $\vdash L:(S,\r)$ such that $\merge(\trcum(D))$ is a tree.
	Then $L$ is a tree, and is equivalent to $\merge(\trcum(D))$.
\end{lemma}

\begin{proof}
	Induction on the structure of $L$.
	If $L$ was of the form $A\,M_1\,\dots\,M_k$, then in $D$ we would necessarily have a node for the nonterminal $A$, which would imply that $\merge(\trcum(D))$ is not a tree, i.e., it contains a nonterminal.
	Thus $L$ is of the form $a^r\,M_1\,\dots\,M_r$.
	Looking at the type system we notice that $D$ necessarily starts with a sequence of $r$ application rules followed by a single-node derivation for %$r$ nodes using the application rule, behind which we have a node labeled by 
	$\vdash a^r:(S_0,\tau_0)$, and derivations $D_i$ for $\vdash M_i:(S_i,\r)$, for $i\in\{1,\dots,r\}$.
	Recall that $\trcum(D)=(a^0;\trcum(D_1);\dots;\trcum(D_r))$.
	For $i\in\{1,\dots,r\}$ we know that $\merge(\trcum(D_i))$ is a tree; the induction assumption implies that $M_i$ is a tree, and is equivalent to $\merge(\trcum(D_i))$.
	It follows that $L$ is a tree, and is equivalent to $\merge(\trcum(D))$.
\end{proof}

\begin{corollary}\label{cor:sound2}
	Let $D$ be a derivation for $\vdash L:(S,\r)$, where $L$ does not contain the initial nonterminal.
	If a tree $Q$ can be generated by $\Ss'$ from $\merge(\trcum(D))$, then a tree equivalent to $Q$ can be generated by $\Ss$ from $L$.
\end{corollary}

\begin{proof}
	Induction on the smallest length of a sequence of reductions $\merge(\trcum(D))\to_{\Ss'}^* Q$.
	If this length is $0$, we apply Lemma \ref{lem:s-end}.
	Suppose that the length is positive.
	Thanks to Lemma \ref{lem:s-leftmost} we can write $\merge(\trcum(D))\to^\mathit{lf}_{\Ss'}P\to_{\Ss'}^*Q$ (without changing the length of the sequence of reductions).
	Using Lemma \ref{lem:s-step} we obtain a term $L'$ and a derivation $D'$ for $\vdash L':(S,\r)$ such that $L\to_\Ss L'$ and that $\trcum(D')$ is $\merge$-equivalent to $P$.
	The initial nonterminal does not appear in $L'$ since by assumption it does not appear on the right side of any rule.
	Because $P\to_{\Ss'}^* Q$, by Lemma \ref{lem:s-merge-equiv} we also have a sequence of reductions of the same length $\merge(\trcum(D'))\to_{\Ss'}^* Q'$ to some tree $Q'$ equivalent to $Q$;
	to this sequence of reductions we apply the induction assumption.
\end{proof}

\begin{proof}[Proof of Lemma \ref{lem:sound}]
  Let $n=|\Delta|$.  Consider the rule $A_\mathit{init}\!\to\!
A\,e_1^0\,\dots\,e_n^0$ from the initial nonterminal of $\Ss$. 
 Let $D$ be a derivation for
  $\vdash \!A\,e_1^0\,\dots\,e_n^0\!:\!(\Delta,\r)$ that consists of a node
  labeled by $\vdash \!A\!:\!(\emptyset,\tau)$ with
  $\tau=\{(\{e_1^0\},\r)\}\!\to\!\dots\!\to\!\{(\{e_n^0\},\r)\}\!\to\!\r$, and of
  nodes labeled by $\vdash \!e_i^0\!:\!(\{e_i^0\},\r)$ for $i\in\{1,\dots,n\}$,
  joined together by application rules.  
  We see that $\trcum(D)=(A_\tau;e_1^0;\dots;e_n^0)$.
	
	Take a tree $Q$ generated by $\Ss'$.
	Since the only rule of $\Ss'$ from the initial nonterminal is $A_\mathit{init}\to\merge(A_\tau;e_1^0;\dots;e_n^0)$, the tree $Q$ is generated by $\Ss'$ also from $\merge(\trcum(D))$.
	By Corollary \ref{cor:sound2} a tree $Q'$ equivalent to $Q$ can be generated by $\Ss$ from $A\,e_1^0\,\dots\,e_n^0$, and thus also from the initial nonterminal.
\end{proof}

\subsection{Completeness}

The proof of Lemma \ref{lem:compl} is similar to the one of Lemma \ref{lem:sound};
we just need to proceed in the opposite direction. 
Namely, we take a sequence of reductions of $\Ss$ finishing in a
finite tree, and then working from the end of the sequence
we construct backwards a sequence of reductions of $\Ss'$.

There is one additional difficulty that was absent in the previous subsection: we need some kind of uniqueness of derivations.
Indeed, while proceeding forwards from $A\,N_1\,\dots\,N_k$ to $K[N_1/x_1,\dots,N_k/x_k]$, we take a derivation for $N_1$ from the single place where $N_1$ appears in the first term, and we put it in multiple places where $N_1$ appears in the second term.
This time we proceed backwards, so there are multiple places in the
second term where we have a derivation for $N_1$. 
Our type system can accommodate different derivations for the
occurrences of $N_1$ having different types, but for each type we have
to ensure that in different occurrences of $N_1$ with this type the
derivations are  the same. 
Because of that we only consider maximal derivations.

A derivation $D$ is called \emph{maximal} if for every internal node of $D$ the following holds: if the label of this node is $\Gamma\vdash L\,M:(S,\tau)$
and it is possible to derive $\Gamma\vdash M:(\emptyset,\sigma)$ for some $\sigma$, then necessarily this node has a child labeled by $\Gamma\vdash M:(\emptyset,\sigma)$.
The following two lemmata say that it is enough to consider only maximal derivations, and that maximal derivations are unique if we restrict ourselves to labeled types with empty subset of $\Delta$.
We will see later that for other types the multiple occurrence problem
mentioned above does not occur.

\begin{lemma}\label{lem:c-exists-max}
	If $\vdash K:(S,\tau)$ can be derived, then it can be derived by a maximal derivation.
\end{lemma}

\begin{proof}
	Let $\tau=\Lambda_1\!\to\!\dots\!\to\!\Lambda_n\!\to\!\r$
	and suppose $D$ is a derivation for $\vdash K^\alpha:(S,\tau)$.
	We prove a stronger statement: if $T_1,\dots,T_n$ are such
that $\tau'=((\Lambda_1\cup(\{\emptyset\}\times
T_1))\to\!\dots\!\to(\Lambda_n\cup(\{\emptyset\}\times T_n))\to\r)$ is a
type in $\Tt^\alpha$ then there exists a maximal derivation $D'$ for
$\vdash \! K \!:\!(S,\tau')$.
	This is shown by induction on the structure of $K$.
	Surely $K$ is not a variable, as then a type judgment with empty type environment could not be derived.
	If $K$ is a nonterminal, then $S=\emptyset$, and $\vdash \! K \! : \! (S,\tau')$ (for any $\tau'\!\in\!\Tt^\alpha$) can be derived by
        a single-node derivation; this is a maximal derivation.
	If $K$ is a symbol, its sort is $o^n \!\to\! o$; by definition of $\LT^o$ we know that $T_i=\emptyset$ for every $i\in\{1,\dots,n\}$, which implies $\tau'=\tau$.
	Thus $D$ derives $\vdash \!K\!:\!(S,\tau')$ and is maximal, since it consists of a single node.

	Finally, suppose that $K=L\,M$.
	Then in $D$ we have a subtree $D_i$ deriving $\vdash \!L\!:\!(S_0,\Lambda_0\!\to\!\tau)$, and for every $\lambda\in\Lambda_0$ a subtree $D_\lambda$ deriving $\vdash M:\lambda$.
	Let $T_0$ contain those $\sigma$ for which we can derive $\vdash \!M\!:\!(\emptyset,\sigma)$ but $(\emptyset,\sigma)\not\in\Lambda_0$.
	Then by the induction assumption there exists a maximal derivation $D_0'$ for $\vdash \!L\!:\!(S_0,(\Lambda_0\cup(\{\emptyset\}\times T_0))\to\tau')$,
	and for every $\lambda\in(\Lambda_0\cup(\{\emptyset\}\times T_0))$ there exists a maximal derivation $D_\lambda'$ for $\vdash \! M \!:\!\lambda$.
	By composing these derivations together, we obtain a maximal
        derivation $D'$ for $\vdash \! K \!: \!(S,\tau')$: the side condition
        of the application rule still holds since we have added
        only derivations for labeled types of the form $(\es,\sigma)$.
\end{proof}

\begin{lemma}\label{lem:c-unique-max}
	For every type judgment of the form $\Gamma\vdash K:(\emptyset,\tau)$ there exists at most one maximal derivation $D$ deriving it.
\end{lemma}

\begin{proof}
	By induction on the structure of $K$.
	If $K$ is a variable, a symbol, or a nonterminal, then $D$ necessarily consists of a single node labeled by the resulting type judgment, so it is unique.
	Suppose that $K=L\,M$.
	Then below the root of $D$, labeled by $\Gamma\vdash \! K\!:\!(\emptyset,\tau)$, we have a subtree $D_0$ deriving $\Gamma\vdash \!L\!:\!(\emptyset,\{\emptyset\}\times T\to\tau)$, and for every $\sigma\in T$ a subtree $D_\sigma$ deriving $\Gamma\vdash \!M\!:\!(\emptyset,\sigma)$.
	By maximality, whenever we can derive $\Gamma\vdash \!M\!:\!(\emptyset,\sigma)$ for some $\sigma$, there should be a child of the root of $D$ labeled by $\Gamma\vdash \!M\!:\!(\emptyset,\sigma)$, and then $\sigma\in T$.
	This fixes the set $T$, and thus the set of child labels.
%	their order is determined by the fixed order on labeled types.
	The derivations $D_0$ and $D_\tau$ for $\tau\in T$ are unique by the induction assumption.
\end{proof}

After these preparatory results about derivations we come back to our
proof. 
The next lemma deals with the base case: for the last term in a
sequence of reductions in $\Ss$ (this term is a narrow tree) we create an
equivalent term that will be the last term in the corresponding sequence of
reductions in $\Ss'$. 

\begin{lemma}\label{lem:c-start}
	Let $S\subseteq\Delta$, and let $K$ be an $S$-narrow tree.
	Then there exists a maximal derivation $D$ for $\vdash K:(S,\r)$ such that $\merge(\trcum(D))$ is a tree equivalent to $K$.
\end{lemma}

\begin{proof}
	We proceed by induction on the structure of $K$,
	which is necessarily of the form $a^r\,M_1\,\dots\,M_r$.
	If $r=0$, then $S=\{a^0\}$ and we take $D$ to be the single-node derivation for $\vdash \! a^0 \! :\!(\{a^0\},\r)$; we have $\trcum(D)=a^0$.
	Suppose that $r\!\geq \!1$.
	Then $S$ can be represented as a union of disjoint sets $S_1,\dots,S_r$
	s.t. $M_i$ is a $S_i$-narrow tree for each $i\!\in\!\{1,\dots,r\}$.
	By induction, $\forall i\!\in\!\{1,\dots,r\}$ we obtain a maximal derivation $D_i$ for $\vdash \! M_i \!:\!(S_i,\r)$
	s.t. $\merge(\trcum(D_i))$ is a tree equivalent to $M_i$.
	The derivation $D$ is obtained by deriving $\vdash\!
        a^r\!:\!(\emptyset,\{(S_1,\r)\}\!\to\!\dots\!\to\!\{(S_r,\r)\}\!\to\!\r)$ and
        attaching $D_1,\dots,D_r$ using the application rule $r$ times.
	Because the $M_i$'s are of sort $o$, and $\LT^o$ does not contain pairs of the form $(\emptyset,\sigma)$, the definition of maximality requires no additional children for the new internal nodes of $D$, and hence $D$ is maximal.
	Thus $\trcum(D)=(a^0;\trcum(D_1);\dots;\trcum(D_r))$. %, which gives the thesis.
\end{proof}

We now describe what happens during a substitution.
%For a term $N^{\alpha_N}$ let us define $\Lambda_\emptyset(N)$ to be the set of those $(\emptyset,\sigma)\in\LT^{\alpha_N}$ for which $\vdash N:(\emptyset,\sigma)$ can be derived.

\begin{lemma}\label{lem:c-subst}
	Suppose that $D'$ is a maximal derivation for $\Gamma\vdash
        K^{\alpha_K}[N/x^{\alpha_x}]:(S,\tau)$, where $N$ is closed. 
        Let $\Lambda_\emptyset$ be the set of those
          $(\emptyset,\sigma)\in\LT^{\alpha_x}$ for which $\vdash
          N:(\emptyset,\sigma)$ can be derived.
	Then there exists a set $\Lambda\in\LT^{\alpha_x}$, a maximal derivation $D_K$ for $\Gamma'\vdash K:(S,\tau)$ with $\Gamma'=\Gamma\cup\{x:\lambda\mid\lambda\in\Lambda\}$,
	and for each $\lambda\in\Lambda$ a maximal derivation $D_\lambda$ for $\vdash N:\lambda$, such that 
	\begin{enumerate}
	\item
          $\Lambda_\emptyset\subseteq\Lambda$,
	\item	for every $\lambda\!\in\!\Lambda\!\setminus\!\Lambda_\emptyset$ in $D_K$ there is a node labeled by $\Gamma'\vdash \!x\!:\!\lambda$,
	\item	the list $\trcum(D')$ is $\merge$-equivalent to the list\\ $\trcum(D_K)[\eta];\trcum(D_{\lambda_1});\dots;\trcum(D_{\lambda_k})$,
		where\\ $\Lambda=\{\lambda_1,\dots,\lambda_k\}$ with $\lambda_1<\dots<\lambda_k$, and\\ $\eta=(\tr(D_{\lambda_1})/x\restr_{\lambda_1},\dots,\tr(D_{\lambda_k})/x\restr_{\lambda_k})$, and
	\item	if $\alpha_K\neq o$ then also $\tr(D')=\tr(D_K)[\eta]$.
	\end{enumerate}
\end{lemma}

\begin{proof}
	We proceed by induction on the structure of $K$.
	By Lemma~\ref{lem:c-exists-max}, for $\lambda\in\Lambda_\emptyset$,
	there exists a maximal derivation for $\vdash N:\lambda$,
	which is unique by Lemma \ref{lem:c-unique-max}.
	We denote this unique derivation by $D_\lambda$.
	
	We consider three cases.
	First suppose that $K$ is a nonterminal, or a symbol, or a variable other than $x$.
	In this case $K[N/x]=K$.
	We take $\Lambda=\Lambda_\emptyset$, and to obtain $D_K$ we just extend the type environment in the only node of $D'$ by $\{x:\lambda\mid\lambda\in\Lambda\}$.
	Points 1-2 hold trivially.
	For points 3-4 we observe that neither $\tr(D_K)$ nor $\trcum(D_K)$ contains a variable $x\restr_\lambda$ (so the substitution $\eta$ does not change these terms);
	additionally $\trcum(D_\lambda)$ for $\lambda\in\Lambda$ are empty (Lemma \ref{lem:s-empty}).
	
	Next, suppose that $K=x$.
	We take $\Lambda=\Lambda_\emptyset\cup\{(S,\tau)\}$.
	As $D_K$ we take the single-node derivation for $\Gamma'\vdash x:(S,\tau)$,
	and as $D_{(S,\tau)}$ we take $D'$ in which we remove the type environment from every node.
	Since $N$ is closed, $D_{(S,\tau)}$ remains a valid derivation and it remains maximal (when $(S,\tau)\in\Lambda_\emptyset$, we have already defined $D_{(S,\tau)}$ previously, but these two definitions give the same derivation).
	Points 1-2 hold trivially.
	We have $\tr(D')=\tr(D_{(S,\tau)})$ and $\trcum(D')=\trcum(D_{(S,\tau)})$.
	We see that $\trcum(D_K)$ is either an empty list (when $\alpha_K\neq o$) or $\bullet^0$ (when $\alpha_K=o$), so attaching $\trcum(D_K)[\eta]$ does not change the class of $\merge$-equivalence.
	Moreover $\trcum(D_\lambda)$ for $\lambda\in\Lambda_\emptyset$ are empty (Lemma \ref{lem:s-empty}), which gives point 3.
	If $\alpha_K\neq o$, we have $\tr(D_K)[\eta]=x\restr_{(S,\tau)}[\eta]=\tr(D_{(S,\tau)})=\tr(D')$ (point 4).

	Finally suppose that $K=L^{\alpha_L}\,M^{\alpha_M}$, which is a more involved case.
	In $D'$, below its root, we have a subtree $C_0'$ deriving $\Gamma\vdash L[N/x]:(S_0,\{(S_1,\tau_1),\dots,(S_n,\tau_n)\}\to\tau)$, and for each $j\in\{1,\dots,n\}$ a subtree $C_j'$ deriving $\Gamma\vdash M[N/x]:(S_j,\tau_j)$,
	where $(S_1,\tau_1)<\dots<(S_n,\tau_n)$, and $S_0\cap(S_1\cup\dots\cup S_n)=\emptyset$, and $S=S_0\cup\dots\cup S_n$.
	We apply the induction assumption to all these subtrees, obtaining a maximal derivation $C_0$ for $\Gamma\cup\{x:\lambda\mid\lambda\in\Lambda_0\}\vdash L:(S_0,\{(S_1,\tau_1),\dots,(S_n,\tau_n)\}\to\tau)$ 
	and for each $j\in\{1,\dots,n\}$ a maximal derivation $C_j$ for $\Gamma\cup\{x:\lambda\mid\lambda\in\Lambda_j\}\vdash M:(S_j,\tau_j)$,
	and for each $j\in\{0,\dots,n\}$ and $\lambda\in\Lambda_j$ a maximal derivation $D_{j,\lambda}$ for $\vdash N:\lambda$.

	Let $\Lambda=\bigcup_{j\in\{0,\dots,n\}}\Lambda_j$.
	For $\lambda\in\Lambda_\emptyset$ we have already defined $D_\lambda$, and we have $D_\lambda=D_{j,\lambda}$ for every $j\in\{0,\dots,n\}$.
	Recall that for every $\lambda\in\Lambda_j\setminus\Lambda_\emptyset$ there is a node in $C_j$ deriving the labeled type $\lambda$, and hence the set on the first coordinate of $\lambda$ is a subset of $S_j$ (point 2).
	Since the sets $S_j$ are disjoint, for every $\lambda\in\Lambda\setminus\Lambda_\emptyset$ there is exactly one $j$ for which $\lambda\in\Lambda_j$, and we define $D_\lambda$ to be $D_{j,\lambda}$ for this $j$.
	
	We extend the type environment in every node of every $C_j$ to $\Gamma'=\Gamma\cup\{x:\lambda\mid\lambda\in\Lambda\}$,
	and we compose  these derivations into a single derivation $D_K$ for $\Gamma'\vdash K:(S,\tau)$ using the rule for application.
	In order to see that $D_K$ is maximal, take some internal node of $D_K$.
	Suppose first that this node is contained inside some $C_j$ and it is labeled by $\Gamma'\vdash P\,Q$, and it is possible to derive $\Gamma'\vdash Q:(\emptyset,\sigma)$.
	Then it is as well possible to derive $\Gamma\cup\{x:\lambda\mid\lambda\in\Lambda_j\}\vdash Q:(\emptyset,\sigma)$, 
	because $\Lambda\setminus\Lambda_j$ contains only labeled types with nonempty set on the first coordinate and they anyway cannot be used while deriving a labeled type with empty set on the first coordinate.
	Thus by maximality of $C_j$ our node has a child labeled by $\Gamma'\vdash Q:(\emptyset,\sigma)$.
	Next, consider the root of $D_K$, and suppose that it is possible to derive $\Gamma'\vdash M:(\emptyset,\sigma)$.
	Then by Lemma \ref{lem:s-subst} it is as well possible to derive $\Gamma'\vdash M[N/x]:(\emptyset,\sigma)$, 
	so also $\Gamma\vdash M[N/x]:(\emptyset,\sigma)$ (since $x$ does not appear in $M[N/x]$),
	which by maximality of $D'$ means that $(\emptyset,\sigma)$ is one of $(S_j,\tau_j)$, 
	and thus the root of $D_K$ has a child labeled by $\Gamma'\vdash M:(\emptyset,\sigma)$ (created out of the root of $C_j$).
	
	Points 1, 2 follow from the induction assumption.
	It remains to prove points 3, 4.
	Let $\Lambda=\{\lambda_1 < \dots < \lambda_k\}$ %with $\lambda_1<\dots<\lambda_k$,
	and $\eta=(\tr(D_{\lambda_1})/x\restr_{\lambda_1},\dots,\tr(D_{\lambda_k})/x\restr_{\lambda_k})$.
	Similarly, let $\Lambda_j=\{\lambda_{j,1} < \dots < \lambda_{j,k_j}\}$ %with $\lambda_{j,1}<\dots<\lambda_{j,k_j}$,
	and $\eta_j=(\tr(D_{\lambda_{j,1}})/x\restr_{\lambda_{j,1}},\dots,\tr(D_{\lambda_{j,k}})/x\restr_{\lambda_{j,k_j}})$.
	Let us first see that $\tr(D')=\tr(D_K)[\eta]$ (not only if $\alpha_K\neq o$, as in point 4, but also if $\alpha_K=o$).
	By induction we know that $\tr(C_0')=\tr(C_0)[\eta_0]$, as surely $\alpha_L\neq o$. 
	Thus $\tr(C_0')=\tr(C_0)[\eta]$, since $\tr(C_0')$ (hence also $\tr(C_0)[\eta_0]$) does not contain variables $x\restr_\lambda$, so substituting for them does not change anything.
	If $\alpha_M=o$, we simply have $\tr(D')=\tr(C_0')$ and $\tr(D_K)=\tr(C_0)$, so clearly $\tr(D')=\tr(D_K)[\eta]$ holds.
	If $\alpha_M\neq o$, by induction we also know that $\tr(C'_j)=\tr(C_j)[\eta_j]$  $\forall j\in\{1,\dots,n\}$, and thus also $\tr(C'_j)=\tr(C_j)[\eta]$;
	we have $\tr(D')=\tr(C_0')\,\tr(C'_1)\,\dots\,\tr(C'_n)$ and similarly $\tr(D_K)=\tr(C_0)\,\tr(C_1)\,\dots\,\tr(C_n)$, so we also obtain $\tr(D')=\tr(D_K)[\eta]$.
	To show  point 3 we prove that  $\trcum(D')$ is $\merge$-equivalent to the list $\trcum(D_K)[\eta];\trcum(D_{\lambda_1});\dots;\trcum(D_{\lambda_k})$.
	By definition $\trcum(D')$ consists of $\trcum(C_j')$ for $j\in\{0,\dots,n\}$, and if $\alpha_K=o$ then also of $\tr(D')$.
	Similarly, $\trcum(D_K)[\eta]$ equals to
        $\trcum(C_0)[\eta];\dots;\trcum(C_n)[\eta]$, prepended by $\tr(D_K)[\eta]$ if $\alpha_K=o$.
% reformulated to remove overfull
	% We come now to the proof of point 3. Consider the list
        % $\trcum(D')$. Our goal is to show that it is $\merge$-equivalent to the list $(\merge(\trcum(D_K))[\eta];\trcum(D_{\lambda_1});\dots;\trcum(D_{\lambda_k})$.
	% By definition $\trcum(D')$ consists of $\trcum(C_j')$ for $j\in\{0,\dots,n\}$, and if $\alpha_K=o$ then also of $\tr(D')$.
	% Similarly, $\merge(\trcum(D_K))[\eta]$ is $\merge$-equivalent
        % to the list $\merge(\trcum(C_0))[\eta];\dots;\merge(\trcum(C_n))[\eta];P$, where $P=\tr(D_K)[\eta]$ if $\alpha_K=o$ and $P=\top$ otherwise.
	We have already shown that $\tr(D')=\tr(D_K)[\eta]$.
        By the induction assumption, the list $\trcum(C_j')$ is
        $\merge$-equivalent to the list $\trcum(C_j)[\eta_j];\trcum(D_{\lambda_{j,1}});\dots;\trcum(D_{\lambda_{j,k_j}})$ for all $j\in\{0,\dots,n\}$.
	We can replace here $\eta_j$ by $\eta$, since $\trcum(C_j')$
        does not contain variables $x\restr_\lambda$ with
        $\lambda\in\Lambda\backslash \Lambda_j$.
	To finish the proof it is enough to observe that the concatenation of the lists $\trcum(D_{\lambda_{j,1}});\dots;\trcum(D_{\lambda_{j,k_j}})$ for $j\in\{0,\dots,n\}$ is $\merge$-equivalent to $\trcum(D_{\lambda_1});\dots;\trcum(D_{\lambda_k})$.
	Indeed, for $\lambda\in\Lambda_\emptyset$ by Lemma
        \ref{lem:s-empty} $\trcum(D_\lambda)$ is empty, and, as we have already shown, every
        $\lambda\in\Lambda\setminus\Lambda_\emptyset$ belongs to
        exactly one $\Lambda_j$.
\end{proof}

\begin{restatable}{lemma}{lemmacstep}
	\label{lem:c-step}
	Let $D'$ be a maximal derivation for $\vdash L':(S,\r)$, and let $L$ be a term that does not contain the initial nonterminal of $\Ss$ and such that $L\to_\Ss L'$.
	Then there exists a maximal derivation $D$ for $\vdash L:(S,\r)$ and a term $P$ that is $\merge$-equivalent to $\trcum(D')$ and such that $\merge(\trcum(D))\to_{\Ss'}P$.
\end{restatable}
\noindent
The lemma is proved by induction on the structure of $L$; cf. App.~\ref{app:lem:c-step}.
The case when $L$ starts with a nonterminal uses Lemma~\ref{lem:c-subst}.

\begin{corollary}\label{cor:compl2}
	Let $L$ be a term that is of sort $o$ and does not contain the initial nonterminal of $\Ss$, and let $M$ be an $S$-narrow tree generated by $\Ss$ from $L$.
	Then there exists a maximal derivation $D$ for $\vdash L:(S,\r)$ such that a tree equivalent to $M$ can be generated by $\Ss'$ from $\merge(\trcum(D))$.
\end{corollary}

\begin{proof}
	We proceed by induction on the smallest length of the sequence of reductions $L\to_\Ss^* M$.
	If $L=M$, we just apply Lemma \ref{lem:c-start}.
	Suppose that the length is positive, and write $L\to_\Ss L'\to_\Ss^* M$.
	The initial nonterminal does not appear in $L'$ since by assumption it does not appear on the right side of any rule.
	By induction we obtain a maximal derivation $D'$ for $\vdash L':(S,\r)$ such that a tree $Q$ equivalent to $M$ can be generated by $\Ss'$ from $\merge(\trcum(D'))$.
	Then, from Lemma \ref{lem:c-step} we obtain a maximal derivation $D$ for $\vdash L:(S,\r)$ and a term $P$ that is $\merge$-equivalent to $\trcum(D')$ and such that $\merge(\trcum(D))\to_{\Ss'}P$.
	By Lemma \ref{lem:s-merge-equiv} a tree equivalent to $Q$ (and hence to $M$) can be generated by $\Ss'$ from $P$, and hence also from $\merge(\trcum(D))$.
\end{proof}

\begin{proof}[Proof of Lemma \ref{lem:compl}]
	Consider a tree $M$ generated by $\Ss$, and a sequence of reductions of $\Ss$ leading to $M$.
	In the first step the initial nonterminal reduces to $A\,e_1^0\,\dots\,e_{|\Delta|}^0$.
	Corollary \ref{cor:compl2} gives us a derivation $D$ for $\vdash A\,e_1^0\,\dots\,e_{|\Delta|}^0:(\Delta,\r)$ such that $\merge(\trcum(D))$ generates a tree equivalent to $M$.
	Necessarily $\trcum(D)=(A_{\tau_0};e_1^0;\dots;e_{|\Delta|}^0)$, so $\merge(\trcum(D))$ is obtained as the result of the initial rule of $\Ss'$.
\end{proof}

%%% Local Variables:
%%% mode: latex
%%% TeX-master: "main"
%%% End:

\section{Conclusions}
\label{sec:conclusions}

% We have shown decidability of the diagonal problem for languages of finite
% words generated by non-deterministic higher-order recursion schemes. 
% For this we needed to work with schemes generating trees and not only
% words. 
% The first step of the construction is to reduce a scheme to an
% equivalent one, but generating only narrow trees. 
% The second, reduces a scheme
% generating narrow trees to a scheme of a smaller order generating
% equivalent trees; but the latter trees may be not narrow.
% This second step is the main part of the proof. 
% It employs a type system designed to track which
% nullary symbols in a term will eventually end up as leaves in the generated
% tree. 
% We use this information to eliminate applications with an argument
% of type $o$, and hence to reduce the order of the scheme.

This work leaves open the question of the exact complexity of the
diagonal problem. The only known lower bound is given by the
emptiness problem, that is the same as for the model-checking
problem~\cite{kobong11}. Our procedure is probably not
optimal, one of the reasons being the use of reflection in
operation Theorem~\ref{thm:HORS:transd}.

%%% Local Variables:
%%% mode: latex
%%% TeX-master: "main"
%%% End:

%\acks
%Acknowledgments.

% We recommend abbrvnat bibliography style.
%\bibliographystyle{abbrvnat}
\bibliographystyle{abbrv}
\bibliography{local}

\begin{thebibliography}{10}

\bibitem{AbdullaBoassonBouajjani:2001}
P.~A. Abdulla, L.~Boasson, and A.~Bouajjani.
\newblock Effective lossy queue languages.
\newblock In {\em In Proc. of ICALP'01}, LNCS, pages 639--651, 2001.

\bibitem{Aho:JACM:1968}
A.~V. Aho.
\newblock Indexed grammars - an extension of context-free grammars.
\newblock {\em J. ACM}, 15(4):647--671, Oct. 1968.

\bibitem{AsadaKobayashi:ICALP16}
K.~Asada and N.~Kobayashi.
\newblock On word and frontier languages of unsafe higher-order grammars.
\newblock To appear in Proc. of ICALP'16.

\bibitem{BachmeierLuttenbergerSchlund:2015}
G.~Bachmeier, M.~Luttenberger, and M.~Schlund.
\newblock Finite automata for the sub- and superword closure of \mbox{CFL}s:
  Descriptional and computational complexity.
\newblock In {\em In Proc. of LATA'15}, volume 8977 of {\em LNCS}, pages
  473--485, 2015.

\bibitem{BreveglieriCherubiniCitriniCrespi-Reghizzi:Ordered:1996}
L.~Breveglieri, A.~Cherubini, C.~Citrini, and S.~Crespi-Reghizzi.
\newblock Multi-push-down languages and grammars.
\newblock {\em Int. J. Found. Comput. Sci.}, 7(3):253--292, 1996.

\bibitem{broadbent10:_recur_schem_logic_reflec}
C.~H. Broadbent, A.~Carayol, C.-H.~L. Ong, and O.~Serre.
\newblock Recursion schemes and logical reflection.
\newblock In {\em LICS'10}, pages 120--129, 2010.

\bibitem{ClementeParysSalvatiWalukiewicz:FSTTCS:2015}
L.~Clemente, P.~Parys, S.~Salvati, and I.~Walukiewicz.
\newblock Ordered tree-pushdown systems.
\newblock In {\em In Proc. of FSTTCS'15}, volume~45 of {\em LIPIcs}, pages
  163--177, 2015.

\bibitem{tata2007}
H.~Comon, M.~Dauchet, R.~Gilleron, C.~L\"oding, F.~Jacquemard, D.~Lugiez,
  S.~Tison, and M.~Tommasi.
\newblock Tree automata techniques and applications.
\newblock \url{http://www.grappa.univ-lille3.fr/tata}, 2007.

\bibitem{Courcelle:1991}
B.~Courcelle.
\newblock On constructing obstruction sets of words.
\newblock {\em Bulletin of EATCS}, 1991.

\bibitem{CzerwinskiMartensVanRooijenZeitoun:2015}
W.~Czerwi{\'n}ski, W.~Martens, L.~van Rooijen, and M.~Zeitoun.
\newblock A note on decidable separability by piecewise testable languages.
\newblock In {\em FCT'15}, volume 9210 of {\em LNCS}, pages 173--185, 2015.

\bibitem{CzerwinskiMartensRooijenZeitounZetzsche}
W.~Czerwi{\'n}ski, W.~Martens, L.~van Rooijen, M.~Zeitoun, and G.~Zetzsche.
\newblock A characterization for decidable separability by piecewise testable
  languages.
\newblock Submitted, 2015.

\bibitem{io_oi_damm}
W.~Damm.
\newblock The {IO}- and {OI}-hierarchies.
\newblock {\em Theoretical Computer Science}, 20:95--207, 1982.

\bibitem{GruberHozerKutrib:TCS:2007}
H.~Gruber, M.~Holzer, and M.~Kutrib.
\newblock The size of {H}igman--{H}aines sets.
\newblock {\em Theor. Comput. Sci.}, 387(2):167--176, 2007.

\bibitem{DownwardPN:ICALP:2010}
P.~Habermehl, R.~Meyer, and H.~Wimmel.
\newblock The downward-closure of {P}etri net languages.
\newblock In {\em In Proc. of ICALP'10}, volume 6199 of {\em LNCS}, pages
  466--477, 2010.

\bibitem{DBLP:conf/fsttcs/Hague11}
M.~Hague.
\newblock Parameterised pushdown systems with non-atomic writes.
\newblock In {\em {FSTTCS}}, volume~13 of {\em LIPIcs}, pages 457--468, 2011.

\bibitem{DBLP:conf/popl/HagueKO16}
M.~Hague, J.~Kochems, and C.-H.~L. Ong.
\newblock Unboundedness and downward closures of higher-order pushdown
  automata.
\newblock In {\em {Proc. of POPL'16}}, pages 151--163, 2016.

\bibitem{HagueMurawskiOngSerre:Collapsible:2008}
M.~Hague, A.~S. Murawski, C.-H.~L. Ong, and O.~Serre.
\newblock Collapsible pushdown automata and recursion schemes.
\newblock In {\em Proc. of LICS'08}, pages 452--461. IEEE Computer Society,
  2008.

\bibitem{Higman:1952}
G.~Higman.
\newblock Ordering by divisibility in abstract algebras.
\newblock {\em Proc. London Math. Soc.}, s3-2(1):326--336, Jan. 1952.

\bibitem{hofman_et_al:LIPIcs:2015:4987}
P.~Hofman and W.~Martens.
\newblock Separability by short subsequences and subwords.
\newblock In {\em ICDT 2015}, volume~31 of {\em LIPIcs}, pages 230--246, 2015.

\bibitem{kobong11}
N.~Kobayashi and C.-H.~L. Ong.
\newblock Complexity of model checking recursion schemes for fragments of the
  modal mu-calculus.
\newblock {\em Logical Methods in Computer Science}, 7(4), 2011.

\bibitem{kobele2015}
G.~M. Kobele and S.~Salvati.
\newblock The {IO} and {OI} hierarchies revisited.
\newblock {\em Information and Computation}, 243:205--221, 2015.

\bibitem{Mayr:TCS:2003}
R.~Mayr.
\newblock Undecidable problems in unreliable computations.
\newblock {\em Theor. Comput. Sci.}, 297(1-3):337--354, Mar. 2003.

\bibitem{Ong:LICS:2006}
C.-H.~L. Ong.
\newblock On model-checking trees generated by higher-order recursion schemes.
\newblock In {\em Proc. of LICS'06}, pages 81--90, 2006.

\bibitem{salvati15:_using}
S.~Salvati and I.~Walukiewicz.
\newblock Using models to model-check recursive schemes.
\newblock {\em Logical Methods In Computer Science}, 2015.

\bibitem{LaTorreMuschollWalukiewicz:2015}
S.~L. Torre, A.~Muscholl, and I.~Walukiewicz.
\newblock Safety of parametrized asynchronous shared-memory systems is almost
  always decidable.
\newblock In {\em In Proc. of CONCUR'15}, volume~42 of {\em LIPIcs}, pages
  72--84, 2015.

\bibitem{vanLeeuwen:1978}
J.~van Leeuwen.
\newblock Effective constructions in well-partially-ordered free monoids.
\newblock {\em Discrete Math.}, 21(3):237--252, May 1978.

\bibitem{Zetzsche:ICALP:2015}
G.~Zetzsche.
\newblock An approach to computing downward closures.
\newblock In {\em In Proc. of ICALP'15}, volume 9135 of {\em LNCS}, pages
  440--451, 2015.

\bibitem{Zetzsche:STACS:2015}
G.~Zetzsche.
\newblock Computing downward closures for stacked counter automata.
\newblock In {\em In Proc. of STACS'15}, volume~30 of {\em LIPIcs}, pages
  743--756, 2015.

\end{thebibliography}

\eject
\appendix

\section{Closure under linear transductions and full trio}

In this section we prove that finite tree languages generated by HORSes are closed under \emph{linear} bottom-up tree transductions.
%and \emph{inverse} bottom-up tree transductions.

%A (bottom-up) finite tree automaton is a tuple $(\Sigma, Q, Q_F, \Delta)$,
%where transitions are of the form
%
%\begin{align*}
%	f (p_1 \, \var x_1) \cdots (p_k \, \var x_k) \goesto {} q (f \, \var x_1 \cdots \var x_k)
%\end{align*}
%
An FTT is \emph{complete} if every variable $x_i$ appearing on the left side of any transition also appears in the term $t$ on the right side of the transition,
i.e., no subtree is discarded.
%
%A \emph{homomorphism} is a special case of FTT where there is only one control state.
%A \emph{restriction} is a homomorphism that selects for every symbol $f \, x_1 \cdots x_k$  a subset of $n$ subtrees $g \, x_{i_1} \cdots x_{i_n}$
%with $1 \leq i_1 < \cdots < i_n \leq k$.
A \emph{restriction} is a special case of an FTT where there is only one control state, and where every transition is of the form 
$a^r\,(q_,x_1)\,\dots\,(q,x_r)\goesto{}q,b^n\,x_{i_1}\,\dots\,x_{i_n}$ with $1 \leq i_1 < \cdots < i_n \leq r$,
i.e., it relabels the tree and discards some its subtrees.
Clearly, every FTT is the composition of a complete FTT with a restriction.

A \emph{higher-order recursion scheme with states} (HORSS)
is a triple $\Hh = (Q, (q_\mathit{init}, A_\mathit{init}), \Rr)$,
where $Q$ is a finite set of control states,
$(q_\mathit{init}, A_\mathit{init})$ is the \emph{initial process}
with $q_\mathit{init}$ the \emph{initial control state} and $A_\mathit{init}$ the \emph{initial nonterminal} that is of sort $o$,
and $\Rr$ is a finite set of rules of the form
\begin{align*}
	\textrm{(I)}	&&
		p, A^{\alpha_1\to\dots\to\alpha_k\to o}\,x_1^{\alpha_1}\,\cdots\,x_k^{\alpha_k} &\to q, K^o \\
	\textrm{(II)}	&&
		p, a^r \, x_1^o \, \cdots \, x_r^o &\to a^r \, (p_1, x_1) \, \cdots \, (p_r, x_r)
\end{align*}
where the term $K$ uses only variables from the set $\{x_1^{\alpha_1},\dots,x_k^{\alpha_k}\}$.
Rules of type (I) are as in standard HORS except that they are guarded by control states.
Rules of type (II) correspond to a finite top-down tree automaton reading the tree produced by the HORS.
The order of $\Ss$ is defined as the highest order of a nonterminal for which there is a rule in $\Ss$.
%
%\begin{align*}
%	\ell, A \, x_1 \cdots x_k \to m, K
%	\ell, A \, x_1 \cdots x_k \to a \, (\ell_1, K_1) \cdots (\ell_n, K_n)
%\end{align*}
%where $\ell$ and $m$ are control states,
%and $A$ is a nonterminal and $K$ a term are as in standard HORS.
%The notion of order of $\Ss$ and the set of rules $\Rr(\Ss)$  is as expected.
%
%Observe that our schemes are \emph{non-deterministic} in the sense that
%$\Rr(\Ss)$ can have many rules with the same  
%nonterminal and control location on the left-hand side. A scheme with at most one rule for
%each nonterminal and control location is called \emph{deterministic}.
%
Let us now describe the dynamics of HORSSes.
A \emph{process} is a pair $(p, M)$ where $M$ is a closed term of sort $o$ and $p$ is a state in $Q$.
%Substitution is extended to processes by \[ (M, q)[K/x] = (M[K/x], q) \] if $K$ is a term.
%
%Substitution is defined as expected:
%\begin{mathpar}
%	A[M/x]=A,\and
%	a[M/x]=a,\and
%	x[M/x]=M,\and
%	y[M/x]=y\mbox{ if }y\neq x,\and
%	(K\,L)[M/x]=K[M/x]\,L[M/x].
%\end{mathpar}
%We shall use the substitution only when $M$ is closed, so there is no need of performing $\alpha$-conversion.
%We also allow simultaneous substitutions: we write
%$K[M_1/x_1,\dots,M_k/x_k]$ to denote the simultaneous substitution of
%$M_1$, \dots, $M_k$ respectively for $x_1$, \dots, $x_k$.
%We notice that when the terms $M_i$ are closed, this amounts to apply
%the substitutions $[M_i/x_i]$ (with $i\in\{1,\dots,k\}$) in any order.
%the order of the pairs $M_i/x_i$ is irrelevant.
%
A \emph{process tree} is a tree built of symbols and processes, where the latter are seen as symbols of rank $0$.
A HORSS $\Hh$ defines a reduction relation $\to_\Hh$ on process trees:
\begin{mathpar}
	\inferrule
		{(p, A\,x_1\,\dots\,x_k \to q, K)\in\Rr(\Hh)}
		{(p, A\,M_1\,\dots\,M_k) \to_\Hh (q, K[M_1/x_1,\dots,M_k/x_k])}
	\and
	\inferrule
		{(p, a^r \, x_1 \, \cdots \, x_r \to a \, (p_1, x_1) \, \cdots \, (p_r, x_r)) \in \Rr(\Hh)}
		{(p, a^r \, M_1 \cdots M_r) \to_\Hh a \, (p_1, M_1) \cdots (p_r, M_r)}
	\and
	\inferrule{
			K_l\to_\Hh K_l'\mbox{ for some }l\in\{1,\dots,r\}
		\\
			K_i=K_i'\mbox{ for all }i\neq l
		}{
			a^r\,K_1\,\dots\,K_r\to_\Hh a^r\,K_1'\,\dots\,K_r'
		}
\end{mathpar}
We are interested in finite trees generated by HORSSes.
A process tree $T$ is a \emph{tree} if it does not contain any process.
A HORSS $\Hh$ \emph{generates} a tree $T$ from a process $(p, M)$ if $(p, M)\to_\Hh^* T$.
The language $\lang \Hh$ is the set of trees generated by the initial process $(q_\mathit{init}, A_\mathit{init})$.

%Intuitively, a HORSS generates only those trees that are generated by $\Ss$ and accepted by $\Aa$.
%In other words, we define $\lang{\Ss, \Aa} = \lang \Ss \cap \lang \Aa$.
%
A HORS can be seen as a special case of a HORSS where $Q$ has only one state $\hat p$
with the trivial rule $\hat p, a \, x_1 \cdots x_k \to a \, (\hat p, x_1) \cdots (\hat p, x_k)$.
It is well known that this extension does not increase expressive power of HORS,
in the sense that given a HORSS $\Hh$
it is possible to construct a (standard) HORS $\Ss$ of the same order as $\Hh$ (but where the arity of nonterminals is increased)
such that $\lang \Hh = \lang \Ss$ \cite{HagueMurawskiOngSerre:Collapsible:2008}.
However, while combining a HORS with an FTT it is convenient to create a HORSS, as its states can be used to simulate states of the FTT.

On the other hand, it is also useful to have the input HORS in a special normalized form, defined next.
We say that a HORS is \emph{normalized} if every its rule is of the form 
\[ A\, x_1\,\dots\, x_p \to h\, (B_1\, x_1\,\dots\, x_p)\, \dots\, (B_r\, x_1\,\dots\, x_p)\ , \]
where $r \geq 0$,
$h$ is either one of the $x_i$'s, a nonterminal, or a symbol,
and the $B_j$'s are nonterminals.
The arity $p$ may be different in each rule.
We will not detail the rather standard procedure of transforming any HORS into a normalized HORS without increasing the order.
It amounts to splitting every rule into multiple rules, using fresh nonterminals in the cut points.

\begin{lemma}
	\label{lemma:HORS:transd:complete}
	HORSes are affectively closed under complete linear tree transductions.
\end{lemma}

\begin{proof}
	Let $\Ss$ be a HORS and let $\Aa$ be a linear FTT.
	We construct a HORSS $\Hh$ s.t. $\lang \Hh = \trans \Aa (\lang \Ss)$.
	The set of control states of $\Hh$ is taken to be the set of control states of the FTT $\Aa$.
	As noted above, we can assume w.l.o.g.~that $\Ss$ is normalized.	
	
	First, if $\Ss$ contains a rule $A\, \vec x \to h\, M_1 \cdots M_r$ with $h$ not a symbol,
	then $\Hh$ contains the rule $p, A\, \vec x \to p, h\, M_1 \cdots M_r$ for every control state $p$.
	
	Next, for every such rule with $h$ being a symbol $a^r$, and for every transition of $\Aa$ having $a^r$ on the left side, 
	we take to $\Hh$ one rule illustrated by means of a representative example:
	if $\Aa$ contains a transition
	\[ a^2 \, (p_1, x_1) \, (p_2, x_2) \goesto {} p, b^2 \, (c^1 \, x_1) \, x_2 \]
	and $\Ss$ contains a rule
		$A \, \vec y \to a^2 \, (B_1 \, \vec y) \, (B_2 \, \vec y)$,
	then $\Hh$ contains the rule
	\[ p, A \, \vec y \to b^2 \, (c^1 \, (p_1, B_1 \, \vec y))\, (p_2, B_2 \, \vec y) \]
		%\quad \textrm{ if } (B_3 \, \vec y) \in \dom (p_3) \]
	%
	Technically speaking, this is not a HORSS rule,
	but it can be turned into one type (I) rule %(with tests)
	and several type (II) rules by adding new states.
	%This case illustrates how a deleting rule of $\Aa$ requires the use of regular tests in the HORSS.
	
	Finally, we also add rules corresponding to $\varepsilon$-transitions of $\Aa$, what is again defined by an example:
	if $\Aa$ contains a transition
	\[ p , x_1 \goesto {} q, a^1 \, x_1 \]
	then, for every nonterminal $A$ of $\Ss$,
	$\Hh$ contains the rule
	\[ q, A \, \vec y \to a^1 \, (p, A \, \vec y) \]
	%
%	This case illustrates how $\varepsilon$-transitions in $\Aa$ are handled using $\varepsilon$-transitions in $\Cc$.	
	
	The two inclusions needed to show that $\lang \Hh = \trans \Aa (\lang \Ss)$
	can be proved straightforwardly by induction on the length of derivations.
\end{proof}

The difficulty in proving closure under possibly non-complete FTTs
is that when combining a (non-complete) FTT transition of the form e.g.~$a^2 \, (p,x_1) \, (p,x_2) \goesto {} p, b^1 \, x_1$
with a HORS rule of the form e.g.~$A \, \vec y \to a^2 \, (B_1 \, \vec y) \, (B_2 \, \vec y)$,
we cannot simply discard the subterm $B_2 \, \vec y$, but we have to make sure that it generates at least one tree on which the FTT has some run.
While concentrating on closure only under restrictions, one think becomes easier: a restriction has a run almost on every tree.
There is, however, one exception: a restriction $\Aa$ does not have a run on a tree that uses a symbol for which $\Aa$ has no transition.
We deal with this in Lemma \ref{lemma:HORS:restrict-to-symbols}, below.
However, knowing that on every tree there is a run of $\Aa$ is not enough;
we also need to know that $B_2 \, \vec y$ generates at least one tree.
This problem is resolved by Lemma \ref{lemma:HORS:productive}.

\begin{lemma}
	\label{lemma:HORS:restrict-to-symbols}
	For every set of (ranked) symbols $\Theta$ and every HORS $\Ss$ we can build a HORS $\Ss'$ of the same order, such that $\Ll(\Ss')$ contains those trees from $\Ll(\Ss)$ which use only symbols from $\Theta$.
\end{lemma}

\begin{proof}
	We start by assuming w.l.o.g.~that $\Ss$ is normalized.	
	Then, we simply remove from $\Ss$ all rules that use symbols not in $\Theta$.
	Then surely trees in $\Ll(\Ss')$ use only symbols from $\Theta$.
	On the other hand, since $\Ss$ was normalized, every removed rule was of the form 
	$A \, \vec y \to a^r \, (B_1 \, \vec y)\,\dots \, (B_r \, \vec y)$ (with $a^r\not\in\Theta$), 
	so whenever such a rule was used, an $a^r$-labeled node was created.
	In consequence, removing these rules has no influence on generating trees that use only symbols from $\Theta$.
\end{proof}

A HORS $\Ss=(A_\mathit{init},\Rr)$ is \emph{productive}
if, whenever we can reduce $A_\mathit{init}$ to a term $M$ (which may contain nonterminals),
then $M$ can be reduced to some finite tree.
By using the reflection operation~\cite{broadbent10:_recur_schem_logic_reflec},
we can easily turn a HORS into a productive one.
\begin{lemma}
	\label{lemma:HORS:productive}
	For every HORS $\Ss$ we can build a productive HORS $\Ss'$ of the same order generating the same trees.
\end{lemma}

\begin{proof}
	First, we construct a deterministic scheme $\mathcal{T}$ from
	the non-deterministic scheme $\mathcal{S}$.
	To $\Tt$ we will be then able to apply a reflection transformation.
	We use a letter $+$ to eliminate
	non-determinism.
	For every nonterminal $A$ of $\mathcal S$ we collect all its rules:
	$A\, x_1\,\dots\, x_p \to K_1, \dots, A\, x_1\,\dots\, x_p\to K_m$, and add to 
	$\mathcal{T}$ the single rule:
	$$A\,x_1\,\dots\, x_p \to +^2\, K_1\, (+^2\, K_2\, (\dots\, (+^2\, K_{m-1}\, K_m)\dots))\,.$$ 
	The (possibly infinite) tree generated by $\mathcal{T}$ represents the
	language of trees generated from $\mathcal{S}$ since the
	non-deterministic choices that can be made in $\mathcal{S}$ are
	represented by nodes labeled by $+$ in the tree generated by
	$\Tt$.
	In this latter tree, we can find every tree generated by $\Ss$ using a
	finite number of rewriting steps consisting of replacing a subtree
	rooted in $+$ by one of its children. 

	We now take the monotone applicative structure
	(see~\cite{salvati15:_using,kobele2015}) $\mathcal{M} =
	(\mathcal{M}_\alpha)_{\alpha\in\mathrm{Sorts}}$ where  $\mathcal{M}_o$
	is the two element lattice, with maximal element $\top$ and minimal
	element $\bot$. 
	Intuitively, $\top$ means nonempty language and $\bot$ means empty language.
	We interpret $+^2$ as the join (max) of its arguments, and every other symbol
	$a^r$ as the meet (min) of its arguments; in particular symbols of rank $0$ are
	interpreted as $\top$.
	This allows us to define the semantics $\sem{M,\chi,\nu}$ of a
	 term given a valuation $\chi$ for nonterminals and  $\nu$
	for variables (these valuations assign to a variable/nonterminal a value in $\Mm$ of
	an appropriate sort).
	The definition of $\sem{M,\chi,\nu}$ is standard, in particular $\sem{K\,L,\chi,\nu} = \sem{K,\chi,\nu}(\sem{L,\chi,\nu})$.

	The meaning of nonterminals in $\Tt$ is given by the least fixpoint
	computation. 
	For a valuation $\chi$ of the nonterminals of $\mathcal{T}$, we write
	$\mathcal{T}(\chi)$ for the valuation $\chi'$ such that
	$\chi'(A)=\lambda g_1.\cdots.\lambda g_p.\sem{K,\chi,[g_1/x_1,\dots,g_p/x_p]}$
	where $A\, x_1\,\dots\, x_p \to K$ is the rule for $A$ in $\mathcal{T}$.
	Then the meaning of nonterminals is given by the valuation that is the
	least fixpoint of this operator:  $\chi_\Tt=\bigwedge\set{\chi :
	  \Tt(\chi)\subseteq \chi}$. 
	Having $\chi_\Tt$ we can define the semantics of a term $M$ in a valuation
	$\nu$ of its free variables as $\sem{M,\nu}=\sem{M,\chi_\Tt,\nu}$. 

	Least fixed point models of schemes induce an interpretation on 
	infinite trees by finite approximations. An infinite tree has value $\top$
	iff  it represents a non-empty
	language~\cite{kobele2015}.
	The important point is that the semantics of a term and that of the
	infinite tree generated from the term coincide. 

	We can now apply to $\Tt$ the reflection
	operation~\cite{broadbent10:_recur_schem_logic_reflec}
	%\cite{DBLP:conf/fsttcs/Haddad13}
	with respect to the above interpretation  $\mathcal{M}$.
	The result is a scheme $\Tt'$ that generates the same tree as
	$\mathcal{T}$ but where every node is additionally marked by a tuple
	$(a_1,\dots, a_r,b)$ where $a_1$, \dots, $a_r$ is the semantics of the
	arguments of that node (i.e., subtrees rooted at its children) and $b$ is the semantics of the subtree rooted
	at that node.  
	What is important here is that $\mathcal{T'}$ has the same order as
	$\mathcal{T}$ which is the same as that of $\mathcal{S}$. 
	The additional labels allow us to remove unproductive parts of the
	tree generated by $\Tt'$.
	For this we introduce two more nonterminals $\Pi_1$ and
	$\Pi_2$  of sort $o\to o \to o$.
	We then add the rules $\Pi_1\, x_1\, x_2 \to x_1$,
	$\Pi_2\, x_1\, x_2 \to x_2$.
	Now we 
	% in the rules of $\mathcal{T}'$ every occurrence of a
	% constant marked by a tuple $(a_1,\dots,a_r,\bot)$ by $W_r$; we also
	replace every occurrence of $+^2$ labeled by $(\top,\bot,\top)$ by
	$\Pi_1$, and every occurrence of $+^2$ labeled by $(\bot,\top,\top)$ by
	$\Pi_2$. 
	% and $W_r x_1\dots x_r \to W_r x_1\dots x_r$
	%(so that the nonterminals $W_r$ represent divergent computations).
	After these transformations we obtain a scheme $\mathcal{T}''$ 
	generating a tree which contains exactly those nodes of $\mathcal{T}'$
	that are labeled with $(\top,\dots,\top,\top)$.  

	We convert
	$\mathcal{T}''$ into a HORS $\mathcal S'$
	whose language is the same as that of $\mathcal{S}$.
	For this we replace every remaining occurrence of $+^2$ (thus labeled by $(\top,\top,\top)$) by a nonterminal $C$ of sort
	$o\to o\to o$, and we add two rewrite rules $C\, x\, y \to x$ and
	$C\, x\, y\to y$. We also remove the additional labels from symbols.
	By construction, $\mathcal S'$ is productive and $\lang{\Ss'} \subseteq \lang{\Ss}$.
	Moreover, since we only eliminated non-productive nonterminals,
	$\lang{\Ss'} = \lang{\Ss}$.
\end{proof}

\begin{lemma}
	\label{lem:HORS:restriction}
	Let $\Ss$ be a productive HORS, and $\Aa$ a restriction such that for every symbol $a^r$ appearing in any tree generated by $\Ss$ there is a transition of $\Aa$ having $a^r$ on the left side.
	Then we can build a HORS $\Ss'$ whose language is $\Tt(\Aa)(\Ll(\Ss))$.
\end{lemma}

\begin{proof}
	First, w.l.o.g.~we assume that $\Ss$ is normalized (notice that while converting a productive HORS to a normalized one, it remains productive).
	Every rule $S\,\vec y\to h\,(B_1\,\vec y)\,\dots\,(B_r\,\vec y)$ of $S$ in which $h$ is not a symbol is also taken to $\Ss'$.
	If $h=a^r$ is a symbol, we consider every transition of $\Aa$ having $a^r$ on the left side.
	Since $\Aa$ is a restriction, this transition is of the form
	\[ a^r\,(p,x_1)\,\dots\,(p,x_r) \goesto {} p,b^n \, x_{i_1} \cdots x_{i_n}\,, \]
	where $1 \leq i_1 < \cdots < i_n \leq r$.
	Then, to $\Ss'$ we take the rule 
	\[ A \, \vec y \to b^n \, (B_{i_1} \, \vec y) \cdots (B_{i_n} \, \vec y)\,. \]
	
	In general, $\trans{\Aa}(\lang \Ss) \subseteq \lang{\Ss'}$.
	Since $\Ss$ is productive,
	the subterms $B_i \, \vec y$ obtained by rewriting the initial nonterminal $A_\mathit{init}$
	produce at least one tree, and since for every symbol in this tree there is a transition of $\Aa$ having this symbol on the left side, $\Aa$ has some run on this tree.
	Thus $\trans{\Aa}(\lang \Ss) = \lang{\Ss'}$.
\end{proof}
\thmtransd*
\begin{proof}
	A transduction $\Aa$ realized by an FTT is the composition of a complete one $\Bb$ and a restriction $\Cc$.
	We first apply Lemma~\ref{lemma:HORS:transd:complete} to the complete transduction realized by $\Bb$.
	Then, using Lemma \ref{lemma:HORS:restrict-to-symbols} we remove from the generated language all trees that use symbols not appearing on the left side of any transition of $\Cc$.
	Next, we turn the resulting HORS into a productive one by Lemma~\ref{lemma:HORS:productive},
	and, finally, we apply Lemma~\ref{lem:HORS:restriction} to the resulting productive HORS and the restriction realized by $\Cc$.
	We end up with a HORS producing the image of $\Aa$ applied to the original HORS, and being of the same order.
\end{proof}

%%% Local Variables:
%%% mode: latex
%%% TeX-master: "main"
%%% End:

\ignore{
	\section{Proof of Lemma~\ref{lem:prod-aut}}\label{app:lem:prod-aut}

	For convenience let us recall the statement of the lemma.

	\lemmaprodaut*
	%\noindent\textbf{Lemma~\ref{lem:prod-aut}.} {\itshape
	%	Let $\Ss$ be a HORS, let $\Aa$ be a nondeterministic finite tree automaton (reading trees generated by $\Ss$), and let $Q'$ be a subset of its set of states.
	%	We can create a HORS $\Ss'$ that is of the same order as $\Ss$ and generates run trees of $\Aa$ on trees generated by $\Ss$ 
		%(trees generated by $\Ss$ labeled additionally by transitions of an accepting run of $\Aa$), 
	%	(trees generated by $\Ss$ with labels replaced by states of an accepting run of $\Aa$), 
	%	restricted to the nodes such that the state in this node and
	 %       in all its ancestors is in $Q'$.}
	%
	
\todo[inline]{For deciding emptiness of HORS shouldn't the citation be Ong LICS'06 instead of \cite{kobele2015}?}

\begin{proof}
From \cite{kobele2015}, we know that we can decide whether a HORS
defines an
empty language. In the case when $\mathcal{S}$ defines an empty
language the lemma trivially holds. We thus assume that $\mathcal{S}$
defines a non-empty language.

The proof of the lemma comprises two steps:
First, we construct a scheme $\mathcal{S}''$ whose language is the language of
run trees of $\mathcal{A}$ over trees in the language of $\mathcal{S}$.
Second, we construct a scheme $\mathcal{S}'$ from
$\mathcal{S}''$ where we keep only the parts of the run trees staying in $Q'$.

Formally, each symbol should have a fixed rank.
Thus, while representing a run in a tree generated by $\Ss'$, for each state $q$ of $\Aa$ and for each $r\in\Nat$ we will use a symbol $(q,r)$ of rank $r$.\footnote{%
	We remark here that in the proof of Corollary \ref{coro:narrowing} we have assumed for simplicity that for each state there is just one symbol.
	Although this is formally incorrect, it is just a technicality to consider separate symbols for different ranks.}
In the auxiliary HORS $\Ss''$ we need to use a richer alphabet than in $\Ss'$: 
%So as to represent a run, 
for each transition $\rho = (a\, q_1\,\dots\, q_r \to q)$ of
$\mathcal{A}$ we will use a symbol $\rho$ of rank $r$. % in $\mathcal{S}'$
%and a constant $\rho'$ of rank $m$ in $\mathcal{S}''$ where $m$ is the
%number of occurrences of states in the sequence $q_1,\dots,q_r$ which
%are in $Q'$.

Before we start our constructions, we assume w.l.o.g.~that 
every rule of $\Ss$ is of the form 
$$A\, x_1\,\dots\, x_p \to h\, (B_1\, x_1\,\dots\, x_p)\, \dots\, (B_r\, x_1\,\dots\, x_p)\,,$$
where $h$ is
either one of the $x_i$'s, a nonterminal, or a symbol. Here,
$r$ can be $0$, and thus we also account for rules
of the form $A\, x_1\,\dots\, x_p\to x_i$, $A\, x_1\,\dots\, x_p \to a$, or
$A\, x_1\,\dots\, x_p\to B$.
The arity $p$ may be different in each rule. We will not detail the
rather standard procedure of transforming a grammar to this form
without increasing the order.

In the sequel of the proof, we assume that $Q = \{q_1,\dots,q_n\}$.

Given a sort $\alpha$, we inductively define the sort $[\alpha]$ as
follows:
$$[o] = o,\,[\alpha\to\beta] = \underbrace{[\alpha] \to \cdots \to
    [\alpha]}_{n}\to [\beta]\,.$$
The set of nonterminals of $\mathcal{S}'$ is the set of pairs
$(A,q)$, noted $A_q$, where $A$ is a nonterminal of $\mathcal{S}$ and
$q$ is in $Q$. Moreover, if $A$ has sort $\alpha$ then $A_q$ has sort
$[\alpha]$. Given a term $M$ made only of nonterminals of
$\mathcal{S}$, we inductively define a term $\langle  M,q
\rangle$ made of nonterminals from $\mathcal{S}'$:
$$
\begin{array}{l@{}l}
\langle A\, M_1\,\dots\, M_p,q \rangle = A_q\,
  & \langle  M_1,q_1 \rangle\, \dots\,
    \langle M_1,q_n \rangle\\ 
  &\phantom{\langle  M_1,q_1 \rangle\, }\dots\\
  &\langle M_p,q_1 \rangle\,\dots\, \langle
    M_p,q_n \rangle\,.
\end{array}
$$
As a shorthand, we write $\overrightarrow{\langle N\rangle}$ to denote a
sequence of arguments of the form  $\langle N,q_1 \rangle\, \dots\, \langle
N,q_n \rangle$. With this notation, we can write the definition of
$\langle M,q \rangle$ as follows:
$$\langle A\, M_1\,\dots\, M_p,q \rangle = A_q\, \overrightarrow{\langle  M_1 \rangle}\, \dots\,
\overrightarrow{\langle M_p \rangle}\,.$$ Similarly we are going to write
$\overrightarrow{x_i}$ to denote a sequence of variables $x_{i,1}\,\dots\,x_{i,n}$.

The rules of $\mathcal{S}''$ are built from the rules of $\mathcal{S}$
as follows:
\begin{itemize}
\item if
  $A\, x_1\,\dots\, x_p \to x_i\, (B_1\,x_1\,\dots\, x_p)\,\dots\, (B_r\,x_1\,\dots\, x_p)$ is
  a rule of $\mathcal{S}$, then, for every $j$ in $\{1,\dots,n\}$, $\mathcal S''$ contains the rule:
  $$
  \begin{array}{l@{}c}
    A_{q_j}\,\overrightarrow{y_1}\,\dots\, \overrightarrow{y_p} \to y_{i,j}\, 
    &
      (B_{1,q_1}\, \overrightarrow{y_1}\,\dots\, \overrightarrow{y_p})\, \dots\,
      (B_{1,q_n}\, \overrightarrow{y_1}\,\dots\, \overrightarrow{y_p})\\ 
    &\dots\\
    & (B_{r,q_1}\,\overrightarrow{y_1}\,\dots\, \overrightarrow{y_p})\, \dots\,
      (B_{r,q_n}\,\overrightarrow{y_1}\,\dots\, \overrightarrow{y_p})
  \end{array}
  $$
\item If
  $A\, x_1\,\dots\, x_p \to B\, (B_1\,x_1\,\dots\, x_p)\,\dots\, (B_r\,x_1\,\dots\, x_p)$ is
  a rule of $\mathcal{S}$, then, for every $j$ in $\{1,\dots,n\}$, $\mathcal S''$ contains the rule:
   $$
  \begin{array}{l@{}c}
    A_{q_j}\, \overrightarrow{y_1}\,\dots\, \overrightarrow{y_p} \to B_{q_j}\,
    & (B_{1,q_1}\, \overrightarrow{y_1}\,\dots\, \overrightarrow{y_p})\, \dots\,
    (B_{1,q_n}\, \overrightarrow{y_1}\,\dots\, \overrightarrow{y_p})\\
    &\dots\\
    & (B_{r,q_1}\, \overrightarrow{y_1}\,\dots\, \overrightarrow{y_p})\,
      \dots\, (B_{r,q_n}\, \overrightarrow{y_1}\,\dots\, \overrightarrow{y_p})
  \end{array}
  $$
\item If
  $A\, x_1\,\dots\, x_p \to a\, (B_1\,x_1\,\dots\, x_p)\,\dots\, (B_r\,x_1\,\dots\, x_p)$ is a
  rule of $\mathcal{S}$ and $\rho = (a\, q_{i_1}\, \dots\, q_{i_r} \to q_j)$ is a
  transition of $\mathcal{A}$, then $\mathcal S''$ contains the rule: %lorenzo{should $\rho$ be $q_j$?}
  $$A_{q_j}\,\overrightarrow{y_1}\,\dots\, \overrightarrow{y_p} \to \rho\,
  (B_{1,q_{i_1}}\,\overrightarrow{y_1}\,\dots\, \overrightarrow{y_p})\, \dots\,
  (B_{r,q_{i_r}}\,\overrightarrow{y_1}\,\dots\, \overrightarrow{y_p})
  $$
\end{itemize}
As the rules of $\mathcal{S}'$ are rather long, we will use a more
compact notation and write
$$A\,\overrightarrow{y_1}\,\dots\, \overrightarrow{y_p} \to h\, (B_{1,q_i}\, \overrightarrow{y_1}\, \dots\,
\overrightarrow{y_p})_{i\in[n]}\,\dots\, (B_{r,q_i}\,\overrightarrow{y_1}\,\dots\,
\overrightarrow{y_p})_{i\in[n]}\,,$$ where $h$ may be either a variable or a
nonterminal.

We are now going to prove that $\mathcal{S}''$ indeed generates the
right set of trees.  For this we show the following property.

\smallskip

% \noindent\textbf{Claim:} given an applicative term $M$ of type $o$ made only of
% nonterminals from $\mathcal{S}$, and a term $t$, we have that
% $M\to_\mathcal{S} t$ iff for every $q_j\in Q$, and every run $u$ of
% $\mathcal{A}$ on $t$ starting with state $q$,
% $\langle M,q_j \rangle \to_{\mathcal{S}'} u$.

% \todo[inline]{Claim reformulated}

\noindent\textbf{Claim.} Let $M$ be a term of sort $o$ made only of
nonterminals of $\mathcal{S}$ and let $t$ be a term of sort $o$ made only of terminals.
Let $u$ be a run tree of $\mathcal{A}$ on $t$ starting from a state $q_j$.
We have
\begin{align*}
	M\to^*_\mathcal{S} t \quad \textrm {iff} \quad
	\langle M,q_j \rangle \to^*_{\mathcal{S}''} u. 
\end{align*}

We first start by the left to right part of the equivalence. We
proceed by induction on the smallest length of a sequence of reductions 
$M\to_\mathcal{S}^* t$. In the sequel we suppose that  $M = A\, M_1\,\dots\,
M_p$.

If $A\, x_1\, \dots\, x_p\to x_i\, (B_1\,x_1\,\dots\, x_p)\,\dots\, (B_r\,x_1\,\dots\, x_p)$
is a rule of $\mathcal{S}$ and
$$M\to_{\mathcal{S}}M_i\,(B_1\, M_1\,\dots\, M_p)\,\dots\, (B_r\, M_1\,\dots\, M_p)
\to^*_{\mathcal{S}} t\,,$$ then, by definition, there is a rule
\begin{align*}&A_{q_j}\, \overrightarrow{y_1}\,\dots\, \overrightarrow{y_p} \to\\&\qquad y_{i,j}\,
(B_{1,q_i}\, \overrightarrow{y_1}\, \dots\, \overrightarrow{y_p})_{i\in[n]}\,\dots\,
(B_{r,q_i}\,\overrightarrow{y_1}\,\dots\, \overrightarrow{y_p})_{i\in[n]}\end{align*} in
$\mathcal{S}''$. As
$$\langle M,q_j \rangle = A_{q_j}\, \langle M_1,q_1 \rangle\, \dots\, \langle
M_1,q_n \rangle\, \dots\, \langle M_p,q_1 \rangle\, \dots\, \langle M_p,q_n
\rangle\,,$$ we have that
$$
\begin{array}{r@{}l}
\langle M,q_j \rangle 
  \to_{\mathcal{S}''}
  & \langle M_i,q_j \rangle\,
\overrightarrow{\langle B_1\, M_1\,\dots\, M_p \rangle}\,\dots\,
\overrightarrow{\langle B_r\, M_1\,\dots\, M_p\rangle}\\
    =\ & \langle M_i\,(B_1\,
M_1\,\dots\, M_p)\,\dots\, (B_r\, M_1\,\dots\, M_p),q_j\rangle\,.
\end{array}
$$ 
Thus, by the induction assumption, as
$$M_i\,(B_1\, M_1\,\dots\, M_p)\,\dots\, (B_r\, M_1\,\dots\, M_p)\to_\mathcal{S}^\ast
t\,,$$
we
have that
$\langle M_i\,(B_1\, M_1\,\dots\, M_p)\,\dots\, (B_r\, M_1\,\dots\, M_p),q_j\rangle
\to_{\mathcal{S}''}^\ast u$ and thus
$\langle M,q_j \rangle \to_{\mathcal{S}''}^\ast u$.

If $A\, x_1\,\dots\, x_p \to B\,(B_1\,x_1\,\dots\, x_p)\,\dots\, (B_r\,x_1\,\dots\, x_p)$ is a
rule of $\mathcal{S}$ and
$$M\to_{\mathcal{S}}B\,(B_1\, M_1\,\dots\, M_p)\,\dots\, (B_r\, M_1\,\dots\, M_p)
\to_{\mathcal{S}}^\ast t\,,$$ then, by definition, there is a rule
\begin{align*}&A_{q_j}\, \overrightarrow{y_1}\,\dots\, \overrightarrow{y_p} \to\\&\qquad B_{q_j}\,
(B_{1,q_i}\, \overrightarrow{y_1}\, \dots\,
\overrightarrow{y_p})_{i\in[n]}\,\dots\, (B_{r,q_i}\,\overrightarrow{y_1}\,\dots\,
\overrightarrow{y_p})_{i\in[n]}\end{align*} in
$\mathcal{S}''$. As
$$\langle M,q_j \rangle = A_{q_j}\, \langle M_1,q_1 \rangle\, \dots\, \langle
M_1,q_n \rangle\, \dots\, \langle M_p,q_1 \rangle\, \dots\, \langle M_p,q_n
\rangle\,,$$ we have that
$$
\begin{array}{r@{}l}
\langle M,q_j \rangle \to_{\mathcal{S}''}&  B_{q_j}\,
\overrightarrow{\langle B_1\, M_1\,\dots\, M_p \rangle}\,\dots\,
\overrightarrow{\langle B_r\, M_1\,\dots\, M_p\rangle}\\
=\ & \langle B\,(B_1\,
M_1\,\dots\, M_p)\,\dots\, (B_r\, M_1\,\dots\, M_p),q_j\rangle\,.
\end{array}
$$ Thus, by the
induction assumption, as
$$B\,(B_1\, M_1\,\dots\, M_p)\,\dots\, (B_r\, M_1\,\dots\, M_p)\to_\mathcal{S}^\ast t\,,$$ we
have that
$\langle B\,(B_1\, M_1\,\dots\, M_p)\,\dots\, (B_r\, M_1\,\dots\, M_p),q_j\rangle
\to_{\mathcal{S}''}^\ast u$ and thus
$\langle M,q_j \rangle \to_{\mathcal{S}''}^\ast u$.

If $A\, x_1\,\dots\, x_p \to a\,(B_1\,x_1\,\dots\, x_p)\,\dots\, (B_r\,x_1\,\dots\, x_p)$ is a
rule of $\mathcal{S}$ and
$$M\to_{\mathcal{S}}a\,(B_1\, M_1\,\dots\, M_p)\,\dots\, (B_r\, M_1\,\dots\, M_p)
\to_{\mathcal{S}}^\ast t\,,$$ then, it must be the case that
$t = a\, t_1\,\dots\, t_r$ and that $u = \rho\, u_1\,\dots\, u_r$ where
$\rho = (a\, q_{i_1}\,\dots\, q_{i_r} \to q_j)$ is a transition of $\mathcal{A}$ and
$u_1, \dots, u_r$ are respectively run trees of $\mathcal{A}$ on $t_1$, \dots,
$t_r$ that respectively start from the states $q_{i_1}, \dots,
q_{i_r}$. Moreover, by definition, there is a rule
$$A_{q_j}\, \overrightarrow{y_1}\,\dots\, \overrightarrow{y_p} \to \rho\,
(B_{1,q_{i_1}}\, \overrightarrow{y_1}\,\dots\, \overrightarrow{y_p})\, \dots\,
(B_{r,q_{i_r}}\, \overrightarrow{y_1}\,\dots\, \overrightarrow{y_p})$$ in
$\mathcal{S}''$. As
\begin{align*}
&\rho\, (B_{1,q_{i_1}}\, \overrightarrow{M_1}\,\dots\, \overrightarrow{M_p})\,\dots\,
(B_{r,q_{i_r}}\, \overrightarrow{M_1}\,\dots\, \overrightarrow{M_p}) \\
&\qquad=\rho\,\langle B_1\,M_1\,\dots\, M_p, q_{i_1} \rangle\,\dots\,\langle B_r\,M_1\,\dots\, M_p, q_{i_r} \rangle\,,
\end{align*}
$\langle M, q_j \rangle \to_{\mathcal{S}''} \rho\, \langle B_1\,M_1\,\dots\,
M_p, q_{i_1} \rangle\, \dots\, \langle B_r\,M_1\,\dots\, M_p, q_{i_r} \rangle$.
Now, for every $k\in\{1,\dots,r\}$, by the induction assumption, as
$B_k\,M_1\,\dots\, M_p\to_\mathcal{S}^\ast t_k$, we have that
$\langle B_k\, M_1\,\dots\, M_p, q_{i_k}\rangle \to_{\mathcal{S}''}^\ast u_k$, and
thus $\langle M,q_j \rangle \to_{\mathcal{S}''}^\ast\rho\, u_1\,\dots\, u_r = u$.

The proof of the right to left implication is very similar to the other direction: it is a simple induction
on the derivation where rules used in $\mathcal{S}''$ are mimicked with
the rules of $\mathcal{S}$ they originate from.

The claim says that for the initial symbol of $\mathcal{S}''$ we can take
$\langle S,q_I\rangle$, where $S$ is the initial symbol of $\Ss$ and
$q_I$ the initial state of $\Aa$.

We now turn to the construction of $\mathcal{S}'$ from
$\mathcal{S}''$. 
For this, we construct a deterministic scheme $\mathcal{T}$ from
the non-deterministic scheme $\mathcal{S}''$.
To $\Tt$ we will be then able to apply a reflection transformation.
We use a symbol $+$ of sort $o\to o\to o$ to eliminate
non-determinism.
For every nonterminal $A$ of $\mathcal{S}''$ we collect all its rules:
$A\, x_1\,\dots\, x_p \to K_1, \dots, A\, x_1\,\dots\, x_p\to K_m$, and add to 
$\mathcal{T}$ the single rule:
$$A\,x_1\,\dots\, x_p \to +\, K_1\, (+\, K_2\, (\dots\, (+\, K_{m-1}\, K_m)\dots)).$$ 

The (possibly infinite) tree generated by $\mathcal{T}$ represents the
language of trees generated from $\mathcal{S}''$ since the
non-deterministic choices that can be made in $\mathcal{S}''$ are
represented by nodes labeled by $+$ in the tree generated by
$\Tt$.
In this latter tree, we can find every tree generated by $\Ss'$ using a
finite number of rewriting steps consisting of replacing a subtree
rooted in $+$ by one of its children.

We now take the monotone applicative structure
(see~\cite{salvati15:_using,kobele2015}) $\mathcal{M} =
(\mathcal{M}_\alpha)_{\alpha\in\mathrm{Sorts}}$ where  $\mathcal{M}_o$
is the two element lattice, with maximal element $\top$ and minimal
element $\bot$. 
We interpret $+$ as the join of its arguments, and every other symbol
$b$ as the meet of its arguments; in particular nullary symbols are
interpreted as $\top$.
This allows us to define the semantics $\sem{M,\chi,\nu}$ of a
 term given a valuation $\chi$ for nonterminals and  $\nu$
for variables (these valuations assign to a symbol a value in $\Mm$ of
an appropriate sort).
The definition of $\sem{M,\chi,\nu}$ is standard, in particular $\sem{K\,L,\chi,\nu} = \sem{K,\chi,\nu}(\sem{L,\chi,\nu})$.

The meaning of nonterminals in $\Tt$ is given by the least fixpoint
computation. 
For a valuation $\chi$ of the nonterminals of $\mathcal{T}$, we write
$\mathcal{T}(\chi)$ for the valuation $\chi'$ such that
$\chi'(A)=\lambda g_1.\cdots.\lambda g_p.\sem{K,\chi,[g_1/x_1,\dots,g_p/x_p]}$
where $A\, x_1\,\dots\, x_p \to K$ is the rule for $A$ in $\mathcal{T}$.
Then the meaning of nonterminals is given by the valuation that is the
least fixpoint of this operator:  $\chi_\Tt=\bigwedge\set{\chi :
  \Tt(\chi)\subseteq \chi}$. 
Having $\chi_\Tt$ we can define the semantics of a term $M$ in a valuation
$\nu$ of its free variables as $\sem{M,\nu}=\sem{M,\chi_\Tt,\nu}$. 

Least fixed point models of schemes induce an interpretation on 
infinite trees by finite approximations. An infinite tree has value $\top$
iff  it represents a non-empty
language~\cite{kobele2015}.
The important point is that the semantics of a term and that of the
infinite tree generated from the term coincide. 

We can now apply to $\Tt$ the reflection
operation~\cite{broadbent10:_recur_schem_logic_reflec}
%\cite{DBLP:conf/fsttcs/Haddad13}
with respect to the above interpretation  $\mathcal{M}$.
The result is a scheme $\Tt'$ that generates the same tree as
$\mathcal{T}$ but where every node is additionally marked by a tuple
$(a_1,\dots, a_r,b)$ where $a_1$, \dots, $a_r$ is the semantics of the
arguments of that node and $b$ is the semantics of the subtree rooted
at that node.  
What is important here is that $\mathcal{T'}$ has the same order as
$\mathcal{T}$ which is the same as that of $\mathcal{S}''$ and
$\mathcal{S}$. 
The additional labels allow us to remove unproductive parts of the
tree generated by $\Tt'$.
For this we introduce two more nonterminals $\Pi_1$ and
$\Pi_2$  of sort $o\to o \to o$.
We then add the rules $\Pi_1\, x_1\, x_2 \to x_1$,
$\Pi_2\, x_1\, x_2 \to x_2$.
Now we 
% in the rules of $\mathcal{T}'$ every occurrence of a
% constant marked by a tuple $(a_1,\dots,a_r,\bot)$ by $W_r$; we also
replace every occurrence of $+$ labeled by $(\top,\bot,\top)$ by
$\Pi_1$, and every occurrence of $+$ labeled by $(\bot,\top,\top)$ by
$\Pi_2$. 
% and $W_r x_1\dots x_r \to W_r x_1\dots x_r$
%(so that the nonterminals $W_r$ represent divergent computations).
After these transformations we obtain a scheme $\mathcal{T}''$ 
generating a tree which contains exactly those nodes of $\mathcal{T}'$
that are labeled with $(\top,\dots,\top,\top)$.  

We convert
$\mathcal{T}''$ into a HORS $\widetilde{\mathcal{S}}$
whose language is the same as that of $\mathcal{S}''$.
For this we replace every remaining occurrence of $+$ (thus labelled by $(\top,\top,\top)$) by a nonterminal $C$ of sort
$o\to o\to o$ and add two rewrite rules $C\, x\, y \to x$ and
$C\, x\, y\to y$. We also  remove the additional labels from symbols.
The important point is that $\widetilde{\mathcal{S}}$
is \emph{productive} in the sense that whenever its initial nonterminal $S$
reduces to a term $M$ (which may contain nonterminals) then $M$
can be reduced to some finite tree.

Now we can obtain $\mathcal{S}'$ from
$\widetilde{\mathcal{S}}$. For this we need to keep only parts of the
run that use states from $Q'$. 
For every symbol 
$\rho = (a\, q_{i_1}\,\dots\, q_{i_r},q)$ let $j_1,\dots,j_k$ be those among $j\in\{1,\dots,r\}$ for which $q_{i_j}$ is in $Q'$.
%$q_{j_1},\dots,q_{j_k}$ be the
%sequence obtained from $q_{i_1}\dots q_{i_r}$ by removing the states
%not in $Q'$. 
%moreover to $\Ss'$ we a
%To $\Ss''$ we add a symbol $\rho'$ of arity $k$, and a
Every use of $\rho$ in a rule of $\widetilde\Ss$ is replaced in $\Ss'$ by a 
nonterminal $R_\rho$ with the associated rule
$R_\rho\, x_1\,\dots\, x_r \to (q,k)\, x_{j_1}\, \dots\, x_{j_k}$ (recall that $(q,k)$ is a symbol of rank $k$).
The scheme $\mathcal{S}'$ generates the required language.

The detour we made through the scheme $\mathcal{T}''$ allows us to
avoid the following pitfall: suppose that $\rho\, M_1\, M_2$ is a term
that can be produced from the initial nonterminal of
$\widetilde{\mathcal{S}}$, it could be the case that in the
transformation from $\widetilde{\mathcal{S}}$ to $\mathcal{S}'$ we
would obtain a term $\rho'\,M'_1$ by mimicking the reductions of
$\widetilde{\mathcal{S}}$. If  $M_2$ were non-productive then this
would introduce an element which does not correspond to an element in the
language of $\widetilde{\mathcal{S}}$. But the detour through
$\mathcal{T}''$ guarantees us that all terms derived in
$\widetilde{\mathcal{S}}$, so also $M_2$, are productive.
% As
% a consequence, we are sure that the trees in the language of
% $\mathcal{S}''$ are trees in the language of $\widetilde{S}$ which
% have been pruned from the subtrees that correspond to runs starting
% from a state that is not in $Q'$. As the language of $\widetilde{S}$
% is the same as the language of $\mathcal{S}'$, the conclusion follows.
%
\end{proof}

%%% Local Variables:
%%% mode: latex
%%% reftex-default-bibliography: ("local.bib")
%%% TeX-master: "main"
%%% End:

}

\section{Proof of Lemma~\ref{lem:c-step}}\label{app:lem:c-step}

We recall the the statement of the lemma.

\lemmacstep*
%\noindent\textbf{Lemma~\ref{lem:c-step}.} {\itshape
% 	Let $D'$ be a maximal derivation for $\vdash L':(S,\r)$, and let $L$ be a term that does not contain the initial nonterminal of $\Ss$ and such that $L\to_\Ss L'$.
%	Then there exists a maximal derivation $D$ for $\vdash
%        L:(S,\r)$ and a term $P$ that is $\merge$-equivalent to
%        $\trcum(D')$ and such that $\merge(\trcum(D))\to_{\Ss'}P$.}
%
\begin{proof}
	We proceed by induction on the structure of $L$.

	Suppose first that $L=a^r\,M_1\,\dots\,M_r$ (where surely $r\geq 1$).
	Then $L'=a^r\,M_1'\,\dots\,M_r'$, where $M_l\to_\Ss M_l'$ for some $l\in\{1,\dots,r\}$, and $M_i=M_i'$ for all $i\neq l$.
	The derivation $D'$ contains a node labeled by $\vdash a^r:(\emptyset,\{(S_1,\r)\}\to\dots\to\{(S_r,\r)\}\to\r)$, and for each $i\in\{1,\dots,r\}$ a subtree $D_i'$ deriving $\vdash M_i':(S_i,\r)$
	(they are merged together by using the application rule $r$ times), where $S_1,\dots,S_r$ are disjoint and their union is $S$.
	We apply the induction assumption to $M_l$, obtaining a derivation $D_l$ for $\vdash M_l:(S_l,\r)$ and a term $P_l$ $\merge$-equivalent to $\trcum(D_l')$ and such that $\merge(\trcum(D_l))\to_{\Ss'}P_l$.
	We can write $P_l=\merge(\lista_l')$ (where the length of $\lista'_l$ and $\trcum(D_l)$ is the same).
	We take $D_i=D_i'$ for $i\neq l$, 
	and out of the single-node derivation for $\vdash a^r:(\emptyset,\{(S_1,\r)\}\to\dots\to\{(S_r,\r)\}\to\r)$ and of derivations $D_i$ for $i\in\{1,\dots,r\}$ we compose a derivation $D$, using the application rule $r$ times.
	We see that $\trcum(D)=(a^0;\trcum(D_1);\dots;\trcum(D_r))$,
        and $\trcum(D')=(a^0;\trcum(D_1');\dots;\trcum(D_r'))$.
        Moreover, taking $\lista'_i=\trcum(D_i)$ for $i\neq l$ we get
	$\merge(\trcum(D))\to_{\Ss'}\merge(a^0;\lista_1';\dots;\lista_r')$, where $\merge(a^0;\lista_1';\dots;\lista_r')$ is $\merge$-equivalent to $\trcum(D')$.
	It remains to observe that $D$ is maximal. 
	Indeed, the nodes inside some $D_i$ have all required children since $D_i$ are maximal, 
	and the new internal nodes created in $D$ describe applications with an argument $M_i$ of sort $o$, and it is impossible to derive $\vdash M_i:(\emptyset,\sigma)$ for any $\sigma$ 
	(since $\LT^o$ does not contain pairs with empty set on the first coordinate).
	
	The remaining possibility is that $L=A\,N_1^{\alpha_1}\,\dots\,N_k^{\alpha_k}$.
	Let $A\,x_1\,\dots\,x_k\to K$ be the rule of $\Ss$ used in the reduction $L\to_\Ss L'$, that is such that $L'=K[N_1/x_1,\dots,N_k/x_k]$.
	Take $D_{0,K}=D'$.
	For $i\in\{1,\dots,k\}$, consecutively, we apply Lemma \ref{lem:c-subst} to $D_{i-1,K}$ and $N_i$, creating sets $\Lambda_i\subseteq\LT^{\alpha_i}$ and maximal derivations $D_{i,K}$ and $D_{i,\lambda}$ for $\lambda\in\Lambda_i$.
	Let $D_K=D_{k,K}$; it derives $\Gamma\vdash K:(S,\r)$, where $\Gamma=\{x_i:\lambda\mid i\in\{1,\dots,k\},\lambda\in\Lambda_i\}$.
	By point 2 of Lemma \ref{lem:c-subst} we know that for every $\lambda\in\Lambda_i$ with a nonempty set on the first coordinate, in $D_K$ there is a node labeled by $\Gamma\vdash x_i:\lambda$.
	On the one hand, since our type systems requires that subsets of $\Sigma_0$ coming from different children are disjoint, we can be sure that the sets on the first coordinate of labeled types in $\Lambda_1,\dots,\Lambda_k$ are disjoint.
	It follows that $\tau_A=\Lambda_1\to\dots\to\Lambda_k\to\r$ is a type.
	On the other hand, nodes labeled by $\Gamma\vdash x_i:\lambda$ give the only possibility for introducing elements of $S$ to our derivation $D_K$ 
	(by assumption in $K$ we do not have nullary symbols, since $A$ is not the initial nonterminal), which means that the union of the sets on the first coordinate of labeled types in $\Lambda_1,\dots,\Lambda_k$ is $S$.
	Since $(S,\r)\in\LT^o$, we have $S\neq\emptyset$, and thus $k\geq 1$, which means that $(\emptyset,\tau_A)$ is a labeled type.
	
	In order to obtain the required derivation $D$ for $\vdash L:(S,\r)$, we start with the single-node derivation for $\vdash A:(\emptyset,\tau_A)$, 
	and using the application rule $k$ times we attach derivations $D_{i,\lambda}$ for each $i\in\{1,\dots,k\}$ and $\lambda\in\Lambda_i$.
	This derivation is maximal, since $D_{i,\lambda}$ were maximal, and by point 1 of Lemma \ref{lem:c-subst} the newly created internal nodes have all required children 
	(whenever it is possible to derive a type judgment $\vdash N_i:(\emptyset,\sigma)$, we are deriving it in $D$).
	
	Recall that $\trcum(D)$ is a concatenation of $\tr(D)$ and of $\trcum(D_{i,\lambda})$ for every $i\in\{1,\dots,k\}$ and $\lambda\in\Lambda_i$.
	For $i\in\{1,\dots,k\}$ let $\eta_i$ be the substitution that maps $x_i\restr_\lambda$ to $\tr(D_{i,\lambda})$ for every $\lambda\in\Lambda_i$.
	In $\Ss'$ we have the rule
        $A_\tau\,\Vars{1}\,\dots\,\Vars{k}\to\merge(\trcum(D_K))$,
        where $\Vars{i}$ lists variables $x_i\restr_{\lambda}$ for
        $\lambda\in\Lambda_i$ if $\alpha_i\neq o$, and is empty if $\alpha_i=o$ (for $i\in\{1,\dots,k\}$).
	Notice that this rule applied to $\tr(D)$ gives $\merge(\trcum(D_K))[\eta_1,\dots,\eta_k]$ 
	(substitutions $\eta_i$ for $i$ such that $\alpha_i=o$ can be skipped, since anyway variables $x_i\restr_{\lambda_{i,j}}$ for such $i$ do not appear in $\trcum(D_K)$).
	As $P$ we take $\merge(\cdot)$ of the concatenation of this term and of all $\trcum(D_{i,\lambda})$; as we have said $\merge(\trcum(D))\to_{\Ss'} P$.
	From point 3 of Lemma \ref{lem:c-subst} it follows that $\trcum(D')$ is $\merge$-equivalent to $P$, what finishes the proof.
\end{proof}

%%% Local Variables:
%%% mode: latex
%%% TeX-master: "main"
%%% End:

%%% Local Variables:
%%% mode: latex
%%% TeX-master: "main"
%%% End:

\end{document}